\documentclass[12pt]{article}
\usepackage{amssymb}
\usepackage{amsthm}
\usepackage{upgreek}
\usepackage[pdfstartview=FitH,
            CJKbookmarks=true,
            bookmarksnumbered=true,bookmarksopen=true,
            colorlinks,
            urlcolor=blue,linkcolor=blue,anchorcolor=blue,citecolor=blue,
            linktocpage=true,
            ]{hyperref}
\usepackage{cite}
\newcommand{\pa}{\partial}

\newcommand{\bbR}{\mathbb{R}}
\newcommand{\bbS}{\mathbb{S}}
\newcommand{\bbZ}{\mathbb{Z}}
\newcommand{\rmd}{{\rm d}}
\newcommand{\fl}{{}}

\theoremstyle{plain}
\newtheorem{theorem}{Theorem}[section]
\newtheorem{lemma}[theorem]{Lemma}

\newtheorem{corollary}[theorem]{Corollary}

\theoremstyle{definition}
\newtheorem{remark}[theorem]{Remark}

\begin{document}

\title{\bf Non-isospectral extension of the Volterra lattice hierarchy, and Hankel determinants}

\author{Xiao-Min Chen$^{1,2}$, Xing-Biao Hu$^{3,4}$ and Folkert M\"uller-Hoissen$^{2,5}$ \\ \small
$^1$ College of Applied Sciences, Beijing University of Technology, \\ \small
Beijing 100124, PR China \\ \small
$^2$  Max Planck Institute for Dynamics and Self-Organization, G\"ottingen, Germany \\ \small
 $^3$ LSEC, Institute of Computational Mathematics and Scientific Engineering Computing, \\ \small
      AMSS, Chinese Academy of Sciences, Beijing 100190, PR China \\ \small
 $^4$ School of Mathematical Sciences, University of Chinese Academy of Sciences, \\ \small
      Beijing 100049, PR China \\ \small
 $^5$ Institute for Nonlinear Dynamics, Georg August University, \\ \small
      Friedrich-Hund-Platz 1, 37077 G\"ottingen, Germany \\ \small
 E-mail: chenxm@lsec.cc.ac.cn, hxb@lsec.cc.ac.cn, folkert.mueller-hoissen@ds.mpg.de     
}


\date{}
\maketitle

\begin{abstract}
For the first two equations of the Volterra lattice hierarchy and the first two equations
of its non-autonomous (non-isospectral) extension, we present Riccati systems for functions $c_j(t)$,
$j=0,1,\ldots$, such that an expression in terms of Hankel determinants built from them
solves these equations on the right half of the lattice. This actually achieves a complete linearization
of these equations of the extended Volterra lattice hierarchy.
\end{abstract}

\section{Introduction}
The Volterra lattice equation\footnote{The subscript of $u_n$ corresponds to a point on a one-dimensional lattice.
We hide away a corresponding subscript of $V^{(1)}$, for simplicity.}
\begin{eqnarray}
    \frac{\rmd u_n}{\rmd t_1} = u_n \left( u_{n+1}-u_{n-1} \right) =: V^{(1)} \, , \qquad \quad  n \in \bbZ  \, ,
        \label{V1}
\end{eqnarray}
is one of the most important integrable partial differential-difference equations \cite{Bogo91}. In particular,
it is a semi-discretization of the inviscid Burgers (also called Riemann or Hopf)
equation, but also of the famous Korteweg-deVries (KdV) equation (cf. \cite{Kac+vanM75,Moser75}).
Because of the latter, it is sometimes referred to as ``discrete KdV equation'' (which should not 
be confused with a full discretization). Another name used for it (mostly if it is presented in a 
certain equivalent form) is ``Kac-van-Moerbeke equation''.
The Volterra lattice is well-known for its use to model population dynamics in biological systems 
\cite{Volt31,Hofb+Sigm98}.
It also models, for example, the propagation of electron density waves (``Langmuir oscillations'') that 
originate from an instability when a periodic electric field is applied to a homogeneous and isotropic 
plasma. More precisely, the Volterra lattice describes the discrete chain of the peaks \cite{ZMR74,Mana75}.  
Moreover, the Volterra lattice also models certain electric network (ladders) built with inductors and 
capacitors \cite{Hiro+Sats76a,Hiro+Sats76b}. 

An integrable partial differential-difference equation typically belongs to a \emph{hierarchy}, which is
an infinite sequence of (somewhat similar) integrable equations, with increasing complexity and such that 
the flows mutually commute.  
The Volterra lattice equation (\ref{V1}) extends to the hierarchy (see, e.g., \cite{OZF89,ZTOF91,Ma+Fuch99}), 
given by 
\begin{eqnarray}
     \frac{\rmd u_n}{\rmd t_k} = V^{(k)} = \mathcal{R}^{k-1} V^{(1)} \, ,  \qquad \quad k = 1,2,\ldots \, ,
        \label{eqV^k}
\end{eqnarray}
with a ``recursion operator'' $\mathcal{R}$ specified in Section~\ref{sec:Volterra_ext}. 
Requiring $u$ to be a common solution, we have 
\begin{eqnarray*}
    \frac{\partial^2 u_n}{\partial t_k \partial t_l} = \frac{\partial^2 u_n}{\partial t_l \partial t_k} 
      \qquad \forall k,l=1,2,\ldots, \; k \neq l \, ,
\end{eqnarray*}
which imposes conditions on the right hand sides of the equations (\ref{eqV^k}). Commutativity of flows means 
that these conditions are fulfilled as a consequence of the equations.  
This ``symmetry condition'' has its roots in Lie theory. Each hierarchy equation is a symmetry of any other, 
see e.g. \cite{Olver86}. If $V'$ denotes the Fr\'echet derivative (see, e.g., \cite{Olver86}) of $V$ 
with respect to $u_n$, and $[V,W] := V'[W]-W'[V]$, the commutativity of flows can be expressed as
\begin{eqnarray*}
    [V^{(k)} , V^{(j)}] = 0  \qquad  j,k=1,2,\ldots \, .
\end{eqnarray*}
The second equation of the Volterra lattice hierarchy is 
\begin{eqnarray}
   \frac{\rmd u_n}{\rmd t_2} = u_n \, \Big( u_{n+1} ( u_n + u_{n+1} + u_{n+2} )
                         - u_{n-1} ( u_{n-2} + u_{n-1} + u_n ) \Big) \, .  \label{V2}
\end{eqnarray}
Explicit expressions of the higher flows of the Volterra lattice hierarchy can also be found in \cite{Svin11,Svin14}.

Moreover, in this work we address an extension of (\ref{eqV^k}) by
additional non-autonomous equations that possess a non-isospectral Lax pair. This means that 
such an equation arises as the compatibility condition of a system of two linear equations (Lax pair), 
depending on a (``spectral'') parameter that is a function of the respective evolution variable.\footnote{Also see
Remark~\ref{rem:integr} below.} 
Because of this property, such equations are often referred to as ``non-isospectral flows''. 
The non-autonomous extension of the Volterra lattice hierarchy, 
\begin{eqnarray}
   \frac{\rmd u_n}{\rmd \tau_k} = \mathcal{V}^{(k)} = \mathcal{R}^k u_n \, ,
   \qquad \quad  k=1,2,\ldots \, , \label{naVeqs}
\end{eqnarray}
first appeared in a different form in \cite{Bere+Shmo94} (also see \ref{app:BS}), 
where more generally a non-autonomous extension of the Toda lattice hierarchy has been treated
(also see \cite{Bere+Shmo90,Levi+Ragn91} and \cite{BGS86,Mokh05,Mokh08}). It was later rediscovered
in \cite{Ma+Fuch99} and generalized to 2+1 dimensions in \cite{GPZ05}, a further
extension, which will not be considered in this work. How (\ref{naVeqs}) indeed leads to 
\emph{non}-autonomous flows is explained in Section~\ref{sec:Volterra_ext}. 
The flows (\ref{naVeqs}) do \emph{not} commute with each other and also not with the flows of the 
Volterra lattice hierarchy. We recall from \cite{Ma+Fuch99} the commutation relations
\begin{eqnarray*}
  [V^{(k)} , \mathcal{V}^{(j)}] = k \, V^{(k+j)} \, , \qquad
  [\mathcal{V}^{(j)},\mathcal{V}^{(k)}] = (k-j) \, \mathcal{V}^{(j+k)} \, , 
\end{eqnarray*}
where $j,k=1,2,\ldots$.
The same structure is also known from non-isospectral extensions of other hierarchies
\cite{Li+Zhu87,Li+Zhu86}. Despite of this non-commutativity of flows, the equation
\begin{eqnarray}
    \frac{\rmd u_n}{\rmd t} = \sum_{k=1}^\infty \Big(\alpha_k(t) \, V^{(k)} + \beta_k(t) \, \mathcal{V}^{(k)} \Big)
                       + \beta_0 \, u_n     \label{total_Volterra_flow}
\end{eqnarray}
is integrable, in the sense of possessing a Lax pair, for any choice of functions $\alpha_k,\beta_k$ of $t$ 
\cite{Bere+Shmo94,Ma+Fuch99}, also see Remark~\ref{rem:integr} below. 

The first non-autonomous equation of the extended hierarchy,
\begin{eqnarray}
  \frac{\rmd u_n}{\rmd \tau_1} = u_n \left[ (n+1) \, u_{n+1} + u_n - (n-2) \, u_{n-1} \right]
                          \, ,    \label{naV1}
\end{eqnarray}
already appeared in \cite{BLR78,Levi+Ragn79,PGW14}. 
Quite surprisingly, this equation is deeply related to the problem of computing the coefficients 
in a continued fraction expansion of a given function, Thiele's expansion formula. Underlying is a method 
to construct interpolation rational functions (Thiele's interpolation formula), see e.g. 
\cite{brezinski1991,cuyt1987nonlinear}. This interesting observation is elaborated in \ref{app:algo}. 

The second non-autonomous equation of the extended hierarchy is already quite involved. 
In Section~\ref{sec:Volterra_ext} we show that it can be chosen as
\begin{eqnarray}
\fl    \frac{\rmd u_n}{\rmd \tau_2} &=& u_n \Big( (3-n) u_{n-2} u_{n-1} +
    (3-n) u_{n-1}^2 + (2-n) u_{n-1} u_n + u_n^2 + (n+3) u_n u_{n+1} \nonumber \\
\fl  && + (n+2) u_{n+1}^2 + (n+2) u_{n+1} u_{n+2}
       + 2 \left( u_{n+1}-u_{n-1} \right) \sum _{i=1}^{n-1} u_i \Big) \, .   \label{naV2}
\end{eqnarray}

Some integrable lattice equations, notably the famous Toda lattice, have been shown to be
linearizable, via Hankel determinants, on the level of time evolution for the entries
of the corresponding Hankel matrices \cite{Bere+Shmo94,KMNOY01,PSZ07,CCSHY15}, also see \ref{app:T->V}. 
A Hankel matrix is a matrix that has the same entries in each ascending skew diagonal. 
There is a deep relation with the
theory of orthogonal polynomials (see, for example, \cite{Chih78} and \cite{Kac+vanM75,Bere+Shmo94,Spice06}),
of which not much is needed, however, to understand the present work. Here we just mention that the entries
of the Hankel matrices have an interpretation as ``moments''. Also see the discussion in
Section~\ref{sec:conclusions}, and \ref{app:T->V}.

For $m=0,1$, and $n \in \bbZ_+$, let
\begin{eqnarray*}
     H^m_n = \det(c_{i+j+m})_{i,j=0}^{n-1}
\end{eqnarray*}
be the determinant of the Hankel matrix with entries $c_{i+j+m}$, $i,j=0,1,\ldots,n-1$. We further set $H^m_0 =1$.
It can be inferred from \cite{Bere+Shmo94} (which actually contains a lot more results) that, for any equation of
the extended Volterra lattice hierarchy, there is a \emph{linear} system of ordinary differential equations
for $c_j$, $j=0,1,\ldots$, such that setting
\begin{eqnarray}
    u_{2n-1} = \frac{H^0_{n-1} \, H^1_n}{H^0_n \, H^1_{n-1}} \, , \qquad
      u_{2n} = \frac{H^0_{n+1} \, H^1_{n-1}}{H^0_n \, H^1_n} \, , \qquad n=1,2,\ldots \, ,
      \label{u<-H}
\end{eqnarray}
yields a solution (of the respective equation of the extended Volterra lattice hierarchy), on the right half
lattice and with boundary condition $u_0=0$. For (\ref{V1}), this is achieved by the simple linear evolution 
equations
\begin{eqnarray*}
      \frac{\rmd c_j}{\rmd t_1} = c_{j+1} \, , \qquad \quad j=0,1,\ldots \, .
\end{eqnarray*}
For the non-autonomous equation (\ref{naV1}), the corresponding system is
\begin{eqnarray*}
      \frac{\rmd c_j}{\rmd \tau_1} = (j+1) \, c_{j+1} \, ,  \qquad \quad j=0,1,\ldots \, .
\end{eqnarray*}
This extends to higher members of the Volterra lattice hierarchy via
\begin{eqnarray*}
      \frac{\rmd c_j}{\rmd t_k} = c_{j+k} \, , \qquad \quad j=0,1,\ldots  \, ,
\end{eqnarray*}
and to members of the non-autonomous extension via
\begin{eqnarray*}
      \frac{\rmd c_j}{\rmd \tau_k} = (j+k) \, c_{j+k} \, ,  \qquad \quad j=0,1,\ldots \, ,
\end{eqnarray*}
where $k=1,2,\ldots$.

The main goal of the present work is to extend the aforementioned results to $u_0 \neq 0$, allowing an
arbitrary function $u_0(t)$.
In this case, the evolution equations for the $c_j$ acquire quadratic terms (see \cite{KMNOY01,PSZ07}
for the case of the Toda equation). This means they have to be extended to a \emph{Riccati system}. For the
Volterra lattice equation (\ref{V1}), this has already been done in \cite{PSZ07}.
In this way, we reach a description of the \emph{complete} set of solutions of the first and second
autonomous, as well as non-autonomous flows, of the extended Volterra lattice hierarchy,
on the right half lattice (with arbitrary boundary data).

The extended Volterra lattice hierarchy is described in more detail in Section~\ref{sec:Volterra_ext}. 
According to our knowledge, this contains the first elaboration of the second non-autonomous 
Volterra lattice equation, moreover with a simplification leading to (\ref{naV2}). 
In Section~\ref{sec:main_results} we collect our main results. Corresponding proofs are then
presented in Section~\ref{sec:proofs}. Section~\ref{sec:conclusions} contains
some additional remarks. \ref{app:T->V} explains the origin of the expressions in (\ref{u<-H}).
\ref{app:BS} recalls a bit of material from \cite{Bere+Shmo94}, and in particular an
expression for the extended Volterra lattice hierarchy, different from that given in \cite{Ma+Fuch99}.
These appendices shall help the reader to quickly understand the relation between (part of) 
\cite{Bere+Shmo94} and \cite{Ma+Fuch99}. 
\ref{app:HBL} addresses a freedom in the (Hirota) bilinearization\footnote{In terms of suitable 
dependent variables, integrable nonlinear equations can typically be cast into a certain bilinear form, 
see \cite{Hiro04}, which makes it easy to access some classes of exact solutions. }
of (\ref{V1}), which is usually disregarded, but important in the present work. This is discussed in
Section~\ref{sec:conclusions}. As already mentioned, \ref{app:algo} establishes a relation between 
the first non-autonomous Volterra flow (\ref{naV1}) and Thiele expansion.

\setcounter{equation}{0}

\section{Volterra lattice hierarchy and a non-autonomous integrable extension}
\label{sec:Volterra_ext}
The recursion operator of the Volterra lattice hierarchy (\ref{eqV^k})
is (see, e.g., \cite{ZTOF91,OZF89,Ma+Fuch99,Saha+Khou03,GPZ05})
\begin{eqnarray*}
     \mathcal{R} = u_n \, (\bbS +1)( u_n \, \bbS - \bbS^{-1} u_n )(\bbS-1)^{-1} u_n^{-1}
                 = H_2 \, H_1^{-1} \, ,
\end{eqnarray*}
with the shift operator $\bbS$ and
\begin{eqnarray*}
  &&  H_1 = u_n \, (\bbS - \bbS^{-1}) \, u_n \, , \\
  &&  H_2 = u_n \, [ (1 + \bbS) \, u_n \, (1 + \bbS) - (1 + \bbS^{-1}) \, u_n \, (1 + \bbS^{-1}) ] \, u_n \, .
\end{eqnarray*}
The latter are Hamiltonian operators \cite{OZF89,Saha+Khou03}.
The operator $\mathcal{R}$ is well-defined on the image of $u_n \, (\bbS-1)$.
Since functions independent of $n$ are in the kernel of the latter,
the result of an application of $\mathcal{R}$ is only determined up to addition of an arbitrary such
function times $V^{(1)}$ (also see \cite{KMW13}). In the following computations we will mostly  
pick a representative and not display this freedom. 

$V^{(1)}$ can be written as
\begin{eqnarray*}
      V^{(1)} = u_n \, ( \bbS - 1 ) (u_n + u_{n-1} ) \, ,
\end{eqnarray*}
so that\footnote{Besides the Volterra lattice equation, also a combination of the first two flows 
of the Volterra lattice hierarchy, $du/dt = V^{(2)} - 6 V^{(1)}$, is an integrable semi-discretization 
of the KdV equation, cf. Section 4.9 in \cite{Suri03}.} 
\begin{eqnarray*}
   V^{(2)} &=& \mathcal{R} V^{(1)} = u_n \, (1+\bbS^{-1})(\bbS \, u_n \, \bbS - u_n)(u_n + u_{n-1}) \nonumber \\
           &=& u_n \, \Big( u_{n+1} ( u_n + u_{n+1} + u_{n+2} )
                         - u_{n-1} ( u_{n-2} + u_{n-1} + u_n ) \Big) \, ,
\end{eqnarray*}
which is the right hand side of (\ref{V2}). 
Since
\begin{eqnarray*}
    V^{(2)} = u_n \, (\bbS -1) \Big( u_n \, ( u_{n-1} + u_n + u_{n+1} )
                         + u_{n-1} ( u_{n-2} + u_{n-1} + u_n ) \Big) \, ,
\end{eqnarray*}
we find
\begin{eqnarray*}
\fl  V^{(3)} &=& \mathcal{R} V^{(2)}  \\
\fl  &=& u_n \, (1+\bbS^{-1})( u_{n+1} \, \bbS^2 - u_n)
       \Big( u_n \, ( u_{n-1} + u_n + u_{n+1} )
                         + u_{n-1} ( u_{n-2} + u_{n-1} + u_n ) \Big) \\
\fl  &=& u_n \Big( u_{n+1} (  u_n^2 + 2 u_n u_{n+1}  + u_{n+1}^2 + u_n u_{n+2} + 2 u_{n+1} u_{n+2} + u_{n+2}^2
                 + u_{n+2} u_{n+3} ) \\
\fl   && \quad - u_{n-1} ( u_{n-2}^2 + 2 u_{n-2} u_{n-1} + u_{n-1}^2 + u_{n-2} u_n + 2  u_{n-1} u_n + u_n^2
                     + u_{n-3} u_{n-2}  )
     \Big)  \, ,
\end{eqnarray*}
and so forth. In these expressions we disregarded possible additional terms arising from the aforementioned indeterminacy
in the action of $\mathcal{R}$, since they simply lead to linear combinations of hierarchy equations.

In \cite{Ma+Fuch99} a non-autonomous extension of the (autonomous) Volterra lattice hierarchy
has been constructed by equations that possess a non-isospectral Lax pair.
A crucial observation is that
$u_n = u_n \, (n - (n-1)) = u_n \, (\bbS -1) (n-1)$ lies in the image of $u_n \, (\bbS-1)$,
and thus
\begin{eqnarray*}
   \mathcal{R} u_n = u_n \Big( (n+1) \, u_{n+1} + u_n - (n-2) \, u_{n-1} \Big) = \mathcal{V}^{(1)}
\end{eqnarray*}
is a local expression. It defines the non-autonomous flow (\ref{naV1}),
which is integrable and possesses a non-isospectral Lax pair. This equation extends to the sequence 
of non-autonomous flows in (\ref{naVeqs}). 
However, for $k>1$, these are non-local expressions. We will discuss this further below.
In (\ref{naVeqs}) we can also allow $k=0$, which is the autonomous linear equation $\rmd u_n/\rmd \tau_0 =u_n$.

\begin{remark}
\label{rem:integr}
Although the non-autonomous Volterra flows do not commute with each other and also not with
the autonomous flows, the combination (\ref{total_Volterra_flow}) of all flows of the 
extended Volterra lattice hierarchy is integrable. The reason for this lies
in the fact that all the flows have one equation of their Lax pair in common \cite{Ma+Fuch99},
namely
\begin{eqnarray*}
     q_{n+1} = \lambda \, q_n - u_n \, q_{n-1} \, ,  \qquad \quad n \in \bbZ \, ,
\end{eqnarray*}
with a (``spectral'') parameter $\lambda$. This equation, restricted to the positive integers $\bbZ_+$,
and with $q_{-1} =0$, $q_0 =1$,
is known to be a recurrence relation for monic symmetric orthogonal polynomials (see, e.g., \cite{Chih78}).
For the Volterra lattice equation (\ref{V1}), we have $\rmd \lambda/\rmd t_1 =0$, and the second part of the Lax pair is
\begin{eqnarray*}
   \frac{\rmd q_n}{\rmd t_1} = - u_{n-1} u_n \, q_{n-2} \, .
\end{eqnarray*}
For the first non-autonomous flow (\ref{naV1}), we have $\rmd \lambda/\rmd \tau_1 = \frac{1}{2} \lambda^3$
and
\begin{eqnarray*}
    \frac{\rmd q_n}{\rmd \tau_1} = \frac{n}{2} q_{n+2}
               + \Big( \sum _{j=1}^{n-1} u_j + \frac{n}{2} (u_n+u_{n+1}) \Big) \, q_n
               - \frac{n}{2} u_{n-1} u_n \, q_{n-2}  \, .
\end{eqnarray*}
\vspace{-1cm}

\hfill $\Box$
\end{remark}

Let us return to the feature of non-locality of the flows (\ref{naVeqs}) for $k>1$. We can write
\begin{eqnarray*}
     \mathcal{V}^{(1)} =  u_n (\bbS -1) \Big( n \, u_n + (n-2) \, u_{n-1} \Big) + 2 \, u_n^2 \, .
\end{eqnarray*}
In order to get a local expression for $\mathcal{V}^{(2)}$, we would need $\mathcal{V}^{(1)}$
to be in the image of $u_n \, (\bbS -1)$, acting on local expressions. But this would mean that
$u_n$ can be expressed as $\bbS -1$ applied to a local term, which is impossible. Hence,
$\mathcal{V}^{(2)}$ contains non-local terms. We obtain
\begin{eqnarray}
     \mathcal{V}^{(2)}
   &=& u_n \, (\bbS +1)( u_n \, \bbS - \bbS^{-1} u_n ) \Big( n \, u_n + (n-2) \, u_{n-1}
       + 2 (\bbS -1)^{-1}(u_n) \Big)  \nonumber  \\
   &=& u_n \Big( (3-n) \, u_{n-2} u_{n-1} + (1-n) \, u_{n-1}^2
       + (2-n) \, u_{n-1} u_n + u_n^2  \nonumber  \\
   &&  + (n+1) \, u_n u_{n+1} + n \, u_{n+1}^2 + (n+2) \, u_{n+1} u_{n+2}
                          \nonumber  \\
   && - 2 u_{n-1} (\bbS -1)^{-1}(u_{n-1}) + 2 u_{n+1} (\bbS -1)^{-1}(u_{n+2}) \Big) \, .  \label{pre_naV2}
\end{eqnarray}
Now we can use
\begin{eqnarray*}
      (\bbS -1)^{-1} = - \sum_{k=0}^\infty \bbS^k
\end{eqnarray*}
to express the nonlocal parts in (\ref{pre_naV2}) in terms of values of $u$ at positive lattice sites only,
\begin{eqnarray*}
\fl     \mathcal{V}^{(2)}
   &=& u_n \Big( (3-n) \, u_{n-2} u_{n-1} + (3-n) \, u_{n-1}^2
       + (2-n) \, u_{n-1} u_n + u_n^2 + (n+3) \, u_n u_{n+1} \nonumber \\
\fl   && + (n+2) \, u_{n+1}^2 + (n+2) \, u_{n+1} u_{n+2} - 2 (u_{n+1}
        - u_{n-1}) \, \sum_{k=n}^\infty u_k \Big) + f \, V^{(1)} \, .
\end{eqnarray*}
Because of the indeterminacy in the action of the recursion operator, mentioned above, here we added
an arbitrary $n$-independent function $f(\tau_2)$ times the right hand side of the Volterra lattice
equation (\ref{V1}). The reason is that, with the choice
\begin{eqnarray*}
     f(\tau_2) = 2 \sum_{k=1}^\infty u_k \, ,
\end{eqnarray*}
the second non-autonomous flow of the extended Volterra lattice hierarchy can be chosen as (\ref{naV2}), 
where the non-locality has been considerably reduced.

\begin{remark}
\label{rem:naV2_BS}
Alternatively we can use
\begin{eqnarray*}
     u_n = (\bbS -1)\Big( \sum_{i=-\infty}^{n-1} u_i \Big)  \, ,
\end{eqnarray*}
to turn (\ref{pre_naV2}) into
\begin{eqnarray*}
\fl      \mathcal{V}^{(2)}
  &=& u_n \Big(
       (3-n) \, u_{n-2} u_{n-1} + (1-n) \, u_{n-1}^2 + (2-n) \, u_{n-1} u_n + 2 u_{n-1} u_{n+1} + u_n^2 \nonumber \\
\fl   && + (n+3) \, u_n u_{n+1} + (n+2) \, u_{n+1}^2 + (n+2) \, u_{n+1} u_{n+2}
       + 2 \left(u_{n+1}-u_{n-1}\right) \sum _{i=-\infty}^{n-2} u_i \Big)  \, .
\end{eqnarray*}
In \cite{Bere+Shmo94}, the boundary condition $u_0=0$ has been chosen. Since
\begin{eqnarray*}
    \sum_{i=-\infty}^{-1} u_i = (\bbS -1)^{-1}(u_0)  \, ,
\end{eqnarray*}
in this case the infinite sum in the above expression for $\mathcal{V}^{(2)}$ reduces to
the finite sum $\sum_{i=1}^{n-2} u_i$. The resulting equation is equivalent to (\ref{naV2}).
\hfill $\Box$
\end{remark}

\setcounter{equation}{0}

\section{Main results}
\label{sec:main_results}

In this section, we state our main results. Proofs are provided in Section~\ref{sec:proofs}.
In an equivalent form, our first theorem has already been stated in \cite{PSZ07}.

\begin{theorem}
\label{thm:V1}
Let $u_0$ be a smooth function of $t_1$. Let $c_j$, $j=0,1,\ldots$, satisfy
\begin{eqnarray}
     \frac{\rmd c_j}{\rmd t_1} = c_{j+1} - \frac{u_0}{c_0} \, \sum_{i=0}^j c_i c_{j-i} \, .
     \label{V1_c_eqs}
\end{eqnarray}
Then (\ref{u<-H}) determines a solution of the Volterra lattice equation (\ref{V1}) on the right
half lattice.
\end{theorem}

\begin{remark}
\label{rem:V1}
Regarding the system (\ref{V1_c_eqs}) as a recurrence relation
\begin{eqnarray*}
   c_{j+1} = \frac{\rmd c_j}{\rmd t_1} + \frac{u_0}{c_0} \, \sum_{i=0}^j c_i c_{j-i} \, ,
\end{eqnarray*}
it determines $c_j(t_1)$, $j>0$, recursively, starting from given functions $u_0(t_1)$ and $c_0(t_1) \neq 0$.
Since
\begin{eqnarray*}
     u_1 = \frac{c_1}{c_0} = \frac{1}{c_0} \frac{\rmd c_0}{\rmd t_1} + u_0 \, ,
\end{eqnarray*}
the theorem yields a solution of (\ref{V1}) on the right half lattice, with arbitrary boundary data
$u_0(t_1)$ and $u_1(t_1)$. Writing the Volterra lattice equation (\ref{V1}) in the form
\begin{eqnarray*}
     u_{n+1} = \frac{1}{u_n} \frac{\rmd u_n}{\rmd t_1} + u_{n-1} \, ,
\end{eqnarray*}
makes evident that the solutions determined by the above system for the $c_j$ then exhaust
the set of its solutions.
\hfill $\Box$
\end{remark}

\begin{theorem}
\label{thm:naV1}
Let $u_0$ be a smooth function of $\tau_1$. Let $c_j$, $j=0,1,\ldots$, satisfy
\begin{eqnarray}
     \frac{\rmd c_j}{\rmd \tau_1} = (j+1) \, c_{j+1} + \frac{u_0}{c_0} \, \sum_{i=0}^j c_i c_{j-i} \, .
      \label{naV1_c_eqs}
\end{eqnarray}
Then (\ref{u<-H}) determines a solution of the first non-autonomous Volterra lattice equation (\ref{naV1})
on the right half lattice.
\end{theorem}

\begin{remark}
Writing the system (\ref{naV1_c_eqs}) as a recurrence relation,
\begin{eqnarray*}
   c_{j+1} = \frac{1}{j+1} \Big( \frac{\rmd c_j}{\rmd \tau_1} - \frac{u_0}{c_0} \, \sum_{i=0}^j c_i c_{j-i} \Big) \, ,
\end{eqnarray*}
it determines $c_j(\tau_1)$, $j>0$, recursively, starting from given functions $u_0(\tau_1)$ and  $c_0(\tau_1)$.
In the same way as in the case of the Volterra lattice equation (\ref{V1}), see Remark~\ref{rem:V1},
Theorem~\ref{thm:naV1} determines a solution of (\ref{naV1}) on the right half lattice,
with arbitrary boundary data $u_0(\tau_1)$ and $u_1(\tau_1)$. Writing (\ref{naV1}) in the form
\begin{eqnarray*}
     u_{n+1} = \frac{1}{n+1} \Big( \frac{1}{u_n} \frac{\rmd u_n}{\rmd \tau_1} - u_n + (n-2) u_{n-1} \Big) \, ,
\end{eqnarray*}
assures that the solutions determined via (\ref{naV1_c_eqs}) comprise all of its solutions.
\hfill $\Box$
\end{remark}

\begin{theorem}
\label{thm:V2}
Let $u_{-1}$ and $u_0$ be smooth functions of $t_2$. Let $c_j$, $j=0,1,\ldots$, satisfy
\begin{eqnarray}
     \frac{\rmd c_j}{\rmd t_2} = c_{j+2} - \frac{u_0}{c_0} \Big( (u_{-1}+u_0) \, \sum_{i=0}^{j-1} c_i c_{j-i}
      + \sum_{i=1}^j c_i c_{j+1-i} \Big) \, .   \label{V2_c_eqs}
\end{eqnarray}
Then (\ref{u<-H}) determines a solution of the second Volterra lattice equation (\ref{V2}) 
on the right half lattice.
\end{theorem}

\begin{remark}
\label{rem:V2}
The system (\ref{V2_c_eqs}), written in the form
\begin{eqnarray*}
   c_{j+2} = \frac{\rmd c_j}{\rmd t_2} + \frac{u_0}{c_0} \Big( (u_{-1}+u_0) \, \sum_{i=0}^{j-1} c_i c_{j-i}
   + \sum_{i=1}^j c_i c_{j+1-i} \Big)\, ,
\end{eqnarray*}
determines $c_j(t_2)$, $j>1$, recursively, starting from given functions $u_{-1}(t_2)$, $u_0(t_2)$,
$c_0(t_2)$ and $c_1(t_2)$. Since $u_1 = c_1/c_0$ and
\begin{eqnarray*}
    u_2 = \frac{c_2}{c_1} - \frac{c_1}{c_0} = \frac{1}{c_1} \frac{\rmd c_0}{\rmd t_2} - u_1 \, ,
\end{eqnarray*}
Theorem~\ref{thm:V2} yields a solution of the second Volterra lattice equation on the right half lattice,
with arbitrary boundary
data $u_{-1}(t_2)$, $u_0(t_2)$, $u_1(t_2)$ and $u_2(t_2)$. Writing (\ref{V2}) in the form
\begin{eqnarray*}
     u_{n+2} = \frac{1}{u_{n+1}} \Big( \frac{1}{u_n}\frac{\rmd u_n}{\rmd t_2} + u_{n-1}(u_{n-2}+u_{n-1}+u_n)
               \Big) - u_n - u_{n+1} \, ,
\end{eqnarray*}
shows that (\ref{V2_c_eqs}) reaches all of its solutions.
\hfill $\Box$
\end{remark}

\begin{theorem}
\label{thm:naV2}
Let $u_{-1}$ and $u_0$ be smooth functions of $\tau_2$. Let $c_j$, $j=0,1,\ldots$, satisfy
\begin{eqnarray}
   \frac{\rmd c_j}{\rmd \tau_2} = (j+2) \, c_{j+2}
     + \frac{u_0}{c_0} \Big( 2 (u_0 + u_{-1}) \sum_{i=0}^{j-1} c_i c_{j-i}
        + \sum_{i=1}^j c_i c_{j+1-i} \Big)  \, .   \label{naV2_c_eqs}
\end{eqnarray}
Then (\ref{u<-H}) determines a solution of the second non-autonomous Volterra lattice equation (\ref{naV2})
on the right half lattice.
\end{theorem}

\begin{remark}
Written in the form
\begin{eqnarray*}
   c_{j+2} = \frac{1}{j+2} \Big[ \frac{\rmd c_j}{\rmd \tau_2}-\frac{u_0}{c_0} \Big( 2 (u_0 + u_{-1}) \sum_{i=0}^{j-1} c_i c_{j-i}
        + \sum_{i=1}^j c_i c_{j+1-i} \Big) \Big] \, ,
\end{eqnarray*}
(\ref{naV2_c_eqs}) determines $c_j(\tau_2)$, $j>1$, recursively, starting from given functions
$u_{-1}(\tau_2)$, $u_0(\tau_2)$, $c_0(\tau_2)$ and $c_1(\tau_2)$.
In the same way as described in Remark~\ref{rem:V2}, we can argue that
Theorem~\ref{thm:naV2} yields a solution of (\ref{naV2}) on the right half lattice, with arbitrary boundary data
$u_{-1}(\tau_2)$, $u_0(\tau_2)$, $u_1(\tau_2)$ and $u_2(\tau_2)$.  Writing (\ref{naV2}) in the form
\begin{eqnarray*}
\fl     u_{n+2}
  &=& \frac{1}{(n+2) \, u_{n+1}} \Big( \frac{1}{u_n}\frac{\rmd u_n}{\rmd \tau_2} - (3-n) u_{n-2} u_{n-1}
      - (3-n) u_{n-1}^2 - (2-n) u_{n-1} u_n - u_n^2  \nonumber \\
\fl  &&  - (n+3) u_n u_{n+1}- (n+2) u_{n+1}^2- 2 \left( u_{n+1}-u_{n-1} \right) \sum _{i=1}^{n-1} u_i  \Big) \, ,
\end{eqnarray*}
shows that the solutions determined by (\ref{naV2_c_eqs}) exhaust the set of its solutions.
\hfill $\Box$
\end{remark}

The above remarks show that the Riccati systems in the theorems are actually equivalent to the respective equations
of the extended Volterra lattice hierarchy.

In the subsequent subsection, following \cite{PSZ07} (also see \cite{Pehe01}) we present a reformulation of the Riccati systems
that appear in the above theorems. According to our knowledge, a Riccati system (beyond the linear case) for
entries of a Hankel determinant first appeared in \cite{KMNOY01}, in the context of the Toda lattice equation.

\subsection{Riccati systems in terms of a Stieltjes function}
Let us introduce the formal series
\begin{eqnarray*}
    \mathcal{F}(\lambda) = \sum_{j=0}^\infty \frac{c_j}{\lambda^{2j+1}} \, ,
\end{eqnarray*}
where $\lambda$ is an indeterminate. This is the ``Stieltjes function'' for symmetric orthogonal polynomials,
a generating function for the moments $s_i$, see (\ref{moments_s}). Since odd moments vanish in this case,
only $c_j := s_{2j}$ shows up. It is easily verified that (\ref{V1_c_eqs}) can be expressed as the
Riccati equation
\begin{eqnarray*}
   \frac{\rmd \mathcal{F}}{\rmd t_1} = - c_0 \lambda  + \lambda^2 \mathcal{F} - \frac{u_0}{c_0} \, \lambda \, \mathcal{F}^2 \, .
\end{eqnarray*}
The latter equation already appeared in \cite{PSZ07}, but with the restriction $c_0 = 1$.

In the same way, (\ref{V2_c_eqs}) can be expressed as the Riccati equation
\begin{eqnarray*}
\fl \frac{\rmd \mathcal{F}}{\rmd t_2} = - c_0 [ \lambda^3 + (u_0+u_1) \lambda ]
   + [ \lambda^4 + 2 u_0 \lambda^2 + (u_{-1}+u_0) u_0 ] \, \mathcal{F}
   - \frac{u_0}{c_0} [ \lambda^3 + (u_{-1}+u_0) \lambda ] \, \mathcal{F}^2 .
\end{eqnarray*}

Allowing $\lambda$ to be a function of $\tau_1$, respectively $\tau_2$, we find that
(\ref{naV1_c_eqs}) and (\ref{naV2_c_eqs}) are equivalent, respectively, to the Riccati equations
\begin{eqnarray*}
   &&\frac{\rmd \mathcal{F}}{\rmd \tau_1} =  -\frac{1}{2} \lambda^2 \, \mathcal{F}
      + \frac{u_0}{c_0} \, \lambda \, \mathcal{F}^2 \, ,  \qquad \mbox{where} \quad
   \frac{\rmd \lambda}{\rmd \tau_1} = \frac{1}{2} \lambda^3 \, ,
\end{eqnarray*}
and
\begin{eqnarray*}
\fl \frac{\rmd \mathcal{F}}{\rmd \tau_2} = c_0 (u_0-u_1) \, \lambda
   -[ \frac{1}{2}\lambda^4 + 2 u_0 \lambda^2 + 2(u_{-1}+u_0) u_0 ] \, \mathcal{F}
    + \frac{u_0}{c_0} [ \lambda^3 + 2(u_{-1}+u_0) \lambda ] \, \mathcal{F}^2 \, ,
\end{eqnarray*}
where
\begin{eqnarray*}
         \frac{d\lambda}{d\tau_2}=\frac{1}{2}\lambda^5 \, .
\end{eqnarray*}

It is well-known that Riccati equations can be linearized. In this way, the above theorems present
linearizations of the respective equations of the extended Volterra lattice hierarchy.
Further analysis can be carried out in analogy to the treatment of a Riccati equation for a
Stieltjes function in Section~6 of \cite{PSZ07}. This also concerns the construction of
special exact solutions.

\setcounter{equation}{0}

\section{Proofs of the theorems}
\label{sec:proofs}
The proofs presented in this section crucially use determinant identities. In particular,
for any determinant $D$, the \emph{Jacobi determinant identity}
\cite{Aitk56,Brua+Schn83} reads
\begin{eqnarray}
 D\cdot
 D\left[\begin{array}{cc}i_1&i_2\\j_1&j_2\end{array}\right]=D\left[\begin{array}{c}i_1\\j_1\end{array}\right]\cdot
 D\left[\begin{array}{c}i_2\\j_2\end{array}\right]-D\left[\begin{array}{c}i_1\\j_2\end{array}\right]\cdot
 D\left[\begin{array}{c}i_2\\j_1\end{array}\right] \, , \label{jacobi}
\end{eqnarray}
where
\begin{eqnarray*}
 D\left[\begin{array}{cccc}i_1&i_2&\cdots&i_k\\j_1&j_2&\cdots&j_k\end{array}\right] \, , \quad
 i_1<i_2<\cdots<i_k \, , \quad j_1<j_2<\cdots<j_k \, ,
\end{eqnarray*}
denotes the determinant obtained from $D$ by removing the rows at positions $i_1,i_2,\ldots,i_k$, and
the columns at positions $j_1,j_2,\ldots,j_k$, in the respective matrix.

In order to demonstrate that a certain Riccati system for $c_j$, $j=0,1,\ldots$, implies that
(\ref{u<-H}) determines solutions of some equation from the extended Volterra lattice hierarchy,
the starting point are the identities
\begin{eqnarray}
\fl \dot{u}_{2n} &=& \left( (\dot{H}_{n+1}^0 H_n^1 - H_{n+1}^0 \dot{H}_n^1) H_n^0 H_{n-1}^1
            - (\dot{H}_{n}^0 H_{n-1}^1 - H_{n}^0 \dot{H}_{n-1}^1) H_{n+1}^0 H_{n}^1 \right)
                  (H_n^0 H_n^1)^{-2}  \, , \nonumber \\
\fl \dot{u}_{2n-1} &=& \left( (H_n^0 \dot{H}_{n}^1 - \dot{H}_n^0 H_{n}^1) H_{n-1}^0 H_{n-1}^1
             - (H_{n-1}^0 \dot{H}_{n-1}^1 - \dot{H}_{n-1}^0 H_{n-1}^1) H_{n}^0 H_{n}^1 \right) (H_n^0 H_{n-1}^1)^{-2} ,
              \nonumber \\
             \label{derivation_u}
\end{eqnarray}
which result from (\ref{u<-H}). Here the ``overdot'' stands for any derivation.

\subsection{Some notation and determinant identities}
\label{subsec:det_ids}
For $m=0,1$, we define
\begin{eqnarray*}
  G^m_n = \left|\begin{array}{cccc}
   c_{m} &c_{m+1} &\cdots &c_{m+n-1} \\
   c_{m+1} &c_{m+2} &\cdots &c_{m+n} \\
   \vdots & \vdots & \ddots & \vdots \\
   c_{m+n-2} & c_{m+n-1} & \cdots & c_{m+2n-3} \\
   c_{m+n} & c_{m+n+1} & \cdots & c_{m+2n-1} \\
    \end{array}\right| ,    \qquad n = 2,3,\ldots \, ,
\end{eqnarray*}
and $G^0_0 =-u_0$, $G^1_0=0$, $G^m_1 = c_{m+1}$.
In Subsections~\ref{subsec:V2} and \ref{subsec:naV2}, we need as well the determinants
\begin{eqnarray*}
 E_n^m &=& \left| \begin{array}{cccc}
c_{m}&\ldots&c_{m+n-2}&c_{m+n+1}\\
c_{m+1}&\ldots&c_{m+n-1}&c_{m+n+2}\\
\vdots&\ddots&\vdots&\vdots\\
c_{m+n-1}&\ldots&c_{m+2n-3}&c_{m+2n}
\end{array}\right| , \qquad n=2,3,\ldots \, , \\
 F_n^m &=&\left|
\begin{array}{ccccc}
c_m&\ldots&c_{m+n-3}&c_{m+n-1}&c_{m+n}\\
c_{m+1}&\ldots&c_{m+n-2} & c_{m+n}&c_{m+n+1}\\
\vdots&\ddots&\vdots&\vdots&\vdots\\
c_{m+n-1}&\ldots&c_{m+2n-4}&c_{m+2n-2}&c_{m+2n-1}
\end{array}\right| ,  \qquad n=3,4,\ldots \, ,
\end{eqnarray*}
in intermediate steps. We set $E^0_0 = -u_0 u_1$, $E^1_0 =0$, $E_1^m = c_{m+2}$, $F^0_0 =u_0 (u_{-1} + u_0 + u_1)$,
$F^1_0 =u_0 u_1$, $F^m_1=0$ and $F^m_2 = H^{m+1}_2$.

In the proofs, presented in the following subsections, we will also use the vectors
\begin{eqnarray}
&&  A_j
 = \left( \begin{array}{c}
 c_j \\
 c_{j+1} \\
  \vdots \\
 c_{j+n-1}
 \end{array}\right) , \quad
 A_j^\ast = \left( \begin{array}{c}
 (n-1) c_j\\
 (n-2) c_{j+1}\\
\vdots\\
 c_{j+n-2}\\
 0
\end{array}\right), \nonumber \\
&& B_j=\left( \begin{array}{c}
 \sum_{i=0}^j c_i c_{j-i} \\
  \sum_{i=0}^{j+1} c_{i} c_{j+1-i} \\
 \vdots\\
  \sum_{i=0}^{j+n-1} c_i c_{j+n-1-i}
 \end{array}\right) , \quad
 \hat{B}_j=\left( \begin{array}{c}
 \sum_{i=1}^j c_i c_{j+1-i} \\
  \sum_{i=1}^{j+1} c_i c_{j+2-i} \\
 \vdots\\
  \sum_{i=1}^{j+n-1} c_i c_{j+n-i}
 \end{array}\right) ,    \label{A,B}
\end{eqnarray}
where, for simplicity, the notation disregards the dependence on a fixed $n$. \\
Note that $|A_m,A_{m+1},\ldots,A_{n-1+m}| = H^m_n$.
The next lemma is only needed for the proof of Lemma~\ref{lem:det_ids}.

\begin{lemma}
 For $m=0,1$, and $n=1,2,\ldots$, the following identities hold,\footnote{The summand in (\ref{sum|AB*A|}) with $j=m$
 is $\left| B_m^\ast, A_{m+1}, \ldots, A_{n-1+m} \right|$. }
\begin{eqnarray}
 && \sum_{j=m}^{n-1+m}\left|A_m,\ldots, A_{j-1}, B_j^\ast, A_{j+1}, \ldots, A_{n-1+m} \right|
  = (n-1) \, c_0 \, H_n^m \, ,  \label{sum|AB*A|} \\
 &&  \sum_{j=m}^{n-1+m}\left|A_m,\ldots, A_{j-1}, \hat{B}_j^\ast, A_{j+1}, \ldots, A_{n-1+m} \right|
  = (n-1) \, c_1 \, H_n^m \, ,  \label{sum|AhatB*A|}
\end{eqnarray}
where
\begin{eqnarray*}
 B_j^\ast=\left( \begin{array}{c}
 0 \\
 c_0 c_{j+1} \\
 \sum_{i=0}^1c_i c_{j+2-i} \\
 \vdots \\
 \sum_{i=0}^{n-2} c_ic_{j+n-1-i}
 \end{array}\right) \, ,  \qquad
 \hat{B}^\ast_j = \left( \begin{array}{c}
 0 \\
 c_1c_{j+1} \\
 \sum_{i=1}^2c_ic_{j+3-i} \\
 \vdots \\
 \sum_{i=1}^{n-1}c_ic_{j+n-i}
\end{array}\right) .
\end{eqnarray*}
\end{lemma}
\begin{proof}
Let $m=0$. For $n=1$, (\ref{sum|AB*A|}) is obviously satisfied. For $n>1$, Laplacian determinant
expansion of the left hand side of (\ref{sum|AB*A|}), with respect to the
$(j+1)$-th column, yields
\begin{eqnarray*}
\fl \sum_{j=0}^{n-1}\left|A_{0},\ldots, A_{j-1},B_j^\ast, A_{j+1}, \ldots, A_{n-1} \right|
    &=& \sum_{j=0}^{n-1}\sum_{k=2}^{n}(-1)^{j+k+1} \sum_{i=0}^{k-2} c_i c_{j+k-1-i} H_{n}^0
     \left[\begin{array}{c} k \\ j+1 \end{array}\right]   \\
\fl &=& \sum_{k=2}^{n}\sum_{i=0}^{k-2} c_i \sum_{j=0}^{n-1}(-1)^{j+k+1} c_{j+k-1-i} H_{n}^0
     \left[\begin{array}{c} k \\ j+1 \end{array}\right] \, .
\end{eqnarray*}
This can be expressed as
\begin{eqnarray*}
\fl \sum_{j=0}^{n-1}\left|A_{0},\ldots, A_{j-1},B_j^\ast, A_{j+1}, \ldots, A_{n-1} \right|
 &=& \sum_{k=2}^{n}\sum_{i=0}^{k-2}c_i\left|
 \begin{array}{cccccccc}
 c_0&c_{1}&\cdots&c_{n-1} \\
 \vdots&\vdots&\ddots &\vdots \\
 c_{k-2}&c_{k-1}&\cdots&c_{k+n-3} \\
 c_{k-1-i}&c_{k-i}&\cdots&c_{k+n-2-i} \\
 c_{k}&c_{k+1}&\cdots&c_{k+n-1} \\
 \vdots&\vdots&\ddots &\vdots \\
 c_{n-1}&c_{n}&\cdots&c_{2n-2}
 \end{array}\right|  \\
\fl  &=& (n-1) c_0 H_{n}^0 \, .
\end{eqnarray*}
In a similar way one proves the other identities.
\end{proof}

\begin{lemma}
\label{lem:det_ids}
For $m=0,1$, and $n=1,2,\ldots$, we have
\begin{eqnarray}
\fl \sum_{j=m}^{n-1+m}\left| A_m,\ldots, A_{j-1}, B_j, A_{j+1}, \ldots, A_{n-1+m} \right|
  &=& (2n-1+m) \, c_0 \, H_n^m \, , \label{sum|ABA|} \\
\fl \sum_{j=m}^{n-1+m}\left|A_m,\ldots, A_{j-1},\hat{B}_j, A_{j+1}, \ldots, A_{n-1+m} \right|
   &=& (2n-2+m) \, c_1 \, H_n^m \, ,  \label{sum|AhatBA|} \\
\fl \sum_{j=m}^{n-1+m}\left| A_m,\ldots, A_{j-1},A_{j+1}^\ast, A_{j+1}, \ldots, A_{n-1+m} \right| &=& 0
     \, , \label{sum|AA*_j+1A|} \\
\fl \sum_{j=m}^{n-1+m}\left|A_m,\ldots, A_{j-1},A_{j+2}^\ast, A_{j+1}, \ldots, A_{n-1+m} \right|
   &=& - F_n^m \, . \label{sum|AA*A|}
\end{eqnarray}
\end{lemma}
\begin{proof} \hspace{1cm} \\
(1) $m=0$. \\
 For $n=1$, (\ref{sum|ABA|}) is a trivial identity.
Let $n>1$. In the determinant on the left hand side of (\ref{sum|ABA|}), to the $(j+1)$-th column,
$j \geq 1$, we add, for $k=1,2,\ldots,j$, the $k$-th column, multiplied by $-c_{j+1-k}$. This yields
\begin{eqnarray*}
\fl \sum_{j=0}^{n-1}\left| A_{0},\ldots, A_{j-1},B_j, A_{j+1}, \ldots, A_{n-1} \right|
 = \sum_{j=0}^{n-1} \left| A_{0},\ldots, A_{j-1},B_j^\ast, A_{j+1}, \ldots, A_{n-1} \right|
   + n c_0 H_n^0  \, .
\end{eqnarray*}
Now (\ref{sum|ABA|}) follows by use of (\ref{sum|AB*A|}).

For $n=1,2$, (\ref{sum|AhatBA|}) is easily verified.
Let $n>2$. In the determinant on the left hand side of (\ref{sum|AhatBA|}), for $j \geq 2$
we add to the $(j+1)$-th column the $k$-th column, multiplied by $-c_{j+2-k}$, for $k=2,3,\ldots,j$,
to obtain
\begin{eqnarray*}
\fl \sum_{j=0}^{n-1}\left|A_{0},\ldots, A_{j-1},\hat{B}_j, A_{j+1}, \ldots, A_{n-1}\right|
 &=& \sum_{j=0}^{n-1} \left|A_{0},\ldots, A_{j-1},\hat{B}_j^\ast, A_{j+1}, \ldots, A_{n-1}\right| \\
 && + (n-1) c_1 H_n^0 \, .
\end{eqnarray*}
Use of (\ref{sum|AhatB*A|}) leads to the stated equation.

(\ref{sum|AA*_j+1A|}) and (\ref{sum|AA*A|}) are easily verified for $n=1$. Let now $n>1$.
(\ref{sum|AA*_j+1A|}) is verified by expanding the determinant on the left hand side with respect to
the $(j+1)$-th column,
\begin{eqnarray*}
\fl \sum_{j=0}^{n-1}\left| A_{0},\ldots, A_{j-1},A_{j+1}^\ast, A_{j+1}, \ldots, A_{n-1} \right|
 &=& \sum_{j=0}^{n-1} \sum_{k=1}^{n-1} (-1)^{k+j+1} (n-k) c_{j+k} H_{n}^0
     \left[\begin{array}{c}k\\j+1\end{array}\right] \\
\fl &=& \sum_{k=1}^{n-1} (n-k) \sum_{j=0}^{n-1}(-1)^{k+j+1} c_{k+j} H_{n}^0
     \left[\begin{array}{c}k\\j+1\end{array}\right] \, ,
\end{eqnarray*}
which vanishes, since the second sum in the last expression is the
Laplace expansion of a determinant having identical $k$-th and $(k+1)$-th rows, for $k=1,\ldots,n-1$.
Addressing (\ref{sum|AA*A|}), we expand the determinant on the left hand side with respect to
the $(j+1)$-th column to obtain
\begin{eqnarray*}
\fl \sum_{j=0}^{n-1}\left|A_{0},\ldots, A_{j-1},A_{j+2}^\ast, A_{j+1}, \ldots, A_{n-1}\right|
 &=& \sum_{j=0}^{n-1} \sum_{k=1}^{n-1} (-1)^{k+j+1} (n-k) c_{k+j+1} H_{n}^0
     \left[\begin{array}{c}k\\j+1\end{array}\right] \\
\fl &=& \sum_{k=1}^{n-1} (n-k) \sum_{j=0}^{n-1} (-1)^{k+j+1} c_{k+j+1} H_{n}^0
     \left[\begin{array}{c}k\\j+1\end{array}\right] \, .
\end{eqnarray*}
Since
\begin{eqnarray*}
   \sum_{j=0}^{n-1}(-1)^{k+j+1}c_{k+j+1}H_{n}^0 \left[\begin{array}{c}k\\j+1\end{array}\right] = 0 \, ,
   \qquad \quad  k=1,2,\ldots,n-2 \, ,
\end{eqnarray*}
which results from a determinant with two equal rows, we obtain
\begin{eqnarray*}
\fl \sum_{j=0}^{n-1} \left|A_0,\ldots, A_{j-1},A_{j+2}^\ast, A_{j+1}, \ldots, A_{n-1}\right|
 = \sum_{j=0}^{n-1} (-1)^{n+j} c_{n+j} H_{n}^0
   \left[\begin{array}{c}n-1\\j+1\end{array}\right]
 = - F_n^0 \, .
\end{eqnarray*}
(2) $m=1$. \\
 For $n=1$, (\ref{sum|ABA|}) is a trivial identity. Let now $n>1$.
In the determinant on the left hand side of (\ref{sum|ABA|}), to the $j$-th
column, $j>1$, we add, for $k = 1,2,\ldots, j-1$, the $k$-th column, multiplied by $-c_{j-k}$, to obtain
\begin{eqnarray*}
\fl && \sum_{j=1}^{n}\left| A_{1},\ldots, A_{j-1},B_j, A_{j+1}, \ldots, A_{n} \right|\nonumber\\
\fl &=& \sum_{j=1}^{n} \Big( c_0H_n^1+\left| A_{1},\ldots, A_{j-1},B_j^\ast, A_{j+1}, \ldots, A_{n} \right|
    + c_j\left| A_{1},\ldots, A_{j-1},A_0, A_{j+1}, \ldots, A_{n} \right| \Big) \, .
\end{eqnarray*}
Then (\ref{sum|ABA|}) follows with the help of (\ref{sum|AB*A|}), and by
Laplace expansion of $\left| A_{1},\ldots, A_{j-1},A_0, A_{j+1}, \ldots, A_{n} \right|$
with respect to the $j$-th column,
\begin{eqnarray*}
\fl && \sum_{j=1}^{n}c_j\left| A_{1},\ldots, A_{j-1},A_0, A_{j+1}, \ldots, A_{n} \right|
    = \sum_{j=1}^{n}c_j\sum_{k=1}^n(-1)^{k+j}c_{k-1}H_n^1\left[\begin{array}{c}k\\j\end{array}\right] \\
\fl   &=& \sum_{k=1}^nc_{k-1}\sum_{j=1}^{n}(-1)^{k+j}c_jH_n^1\left[\begin{array}{c}k\\j\end{array}\right]
    = c_0 \sum_{j=1}^n (-1)^{1+j} c_j H_n^1 \left[\begin{array}{c} 1 \\ j \end{array} \right] = c_0H_n^1 \, ,
\end{eqnarray*}
where we used
\begin{eqnarray*}
 && \sum_{j=1}^{n}(-1)^{k+j}c_jH_n^1\left[\begin{array}{c}k\\j\end{array}\right]=0 \, , \qquad k=2,3,\ldots,n,
\end{eqnarray*}
which results from a determinant with two equal rows.

 For $n=1$, (\ref{sum|AhatBA|}) is easily checked. Let now $n>1$.
In the determinant on the left hand side of (\ref{sum|AhatBA|}), to the $j$-th
column, $j>1$, we add, for $k = 1,2,\ldots, j-1$, the $k$-th column, multiplied by $-c_{j+1-k}$, to obtain
\begin{eqnarray*}
\fl \sum_{j=1}^{n}\left| A_{1},\ldots, A_{j-1},\hat{B}_j, A_{j+1}, \ldots, A_{n} \right|
  = \sum_{j=1}^{n} \Big( c_1 H_n^1 + \left| A_{1},\ldots, A_{j-1},\hat{B}_j^\ast, A_{j+1}, \ldots, A_{n} \right| \Big) \, .
\end{eqnarray*}
Now (\ref{sum|AhatBA|}) follows by use of (\ref{sum|AhatB*A|}).

(\ref{sum|AA*_j+1A|}) and (\ref{sum|AA*A|}), with $m=1$, are proved analogously to the case $m=0$.
\end{proof}

\begin{lemma}
\label{lem:det-identities}
For $n=1,2,\ldots$, the following identities hold,
\begin{eqnarray}
 && G_n^1 H_n^0 - G_n^0 H_n^1 - H_{n+1}^0 H_{n-1}^1 = 0  \, , \label{G^1_identity} \\
 && G_{n}^0 H_{n-1}^1 - G_{n-1}^1 H_{n}^0 - H_{n-1}^0 H_{n}^1 = 0 \, , \label{G^0_identity} \\
 && F_{n}^m H_{n-1}^m + E_{n-1}^m H_{n}^m - G_n^m G_{n-1}^m=0 \, , \quad m=0,1 \, , \label{iden for F^m} \\
 && E_n^1 H_n^0 - E_n^0 H_n^1 = G_{n+1}^0 H_{n-1}^1 \, , \label{identi-E^m-E^m} \\
 && F_n^1 H_n^0 - F_n^0 H_n^1 = G_{n-1}^1 H_{n+1}^0 \, . \label{identi-F^m-F^m}
\end{eqnarray}
\end{lemma}
\begin{proof}
(\ref{G^1_identity}) is true for $n=1$. For $n\geq 2$,
it is the Jacobi identity, see (\ref{jacobi}), with $D = {H}_{n+1}^0$ and
\begin{eqnarray*}
 && D\left[\begin{array}{cc}1&n+1\\n&n+1\end{array}\right]=H_{n-1}^1, \quad
 D\left[\begin{array}{c}1\\n\end{array}\right]=G_n^1, \\
 && D\left[\begin{array}{c}n+1\\n+1\end{array}\right]=H_n^0, \quad
 D\left[\begin{array}{c}1\\n+1\end{array}\right]=H_{n}^1, \quad
 D\left[\begin{array}{c}n+1\\n\end{array}\right]=G_n^0 \, .
\end{eqnarray*}
For $n=1$, (\ref{G^0_identity}) is obvious. For $n \geq 2$ it is obtained as the Jacobi identity with
\begin{eqnarray*}
 D=\left|
\begin{array}{cccccc}
 c_0&c_{1}&\cdots&c_{n-2}&c_{n-1}&c_{n}\\
 c_{1}&c_{2}&\cdots&c_{n-1}&c_{n}&c_{n+1}\\
 \vdots&\vdots&\ddots &\vdots&\vdots&\vdots\\
 c_{n-1}&c_{n}&\cdots&c_{2n-3}&c_{2n-2}&c_{2n-1}\\
 0&0&\cdots&0&1&0
\end{array}\right|=-G_{n}^0
\end{eqnarray*}
and
\begin{eqnarray*}
&&D\left[\begin{array}{cc}n&n+1\\1&n+1\end{array}\right]=H_{n-1}^1,\quad
D\left[\begin{array}{c}n\\1\end{array}\right]=-G_{n-1}^1,\\
&&D\left[\begin{array}{c}n+1\\n+1\end{array}\right]=H_{n}^0, \quad
D\left[\begin{array}{c}n\\n+1\end{array}\right]=H_{n-1}^0,\quad D\left[\begin{array}{c}n+1\\1\end{array}\right]=H_{n}^1.
\end{eqnarray*}
(\ref{iden for F^m}) is easily checked for $n=1,2$, using the definitions of the determinants.
For $n \geq 3$, (\ref{iden for F^m}) is obtained as the Jacobi identity (\ref{jacobi}) with
\begin{eqnarray*}
 D=\left|
\begin{array}{ccccc}
c_m&c_{m+1}&\cdots&c_{m+n-1}&0\\
\vdots&\vdots&\ddots &\vdots&\vdots\\
c_{m+n-3}&c_{m+n-2}&\cdots&c_{m+2n-4}&0\\
c_{m+n-2}&c_{m+n-1}&\cdots&c_{m+2n-3}&1\\
c_{m+n-1}&c_{m+n}&\cdots&c_{m+2n-2}&0\\
c_{m+n}&c_{m+n+1}&\cdots&c_{m+2n-1}&0\\
\end{array}\right|=F_{n}^m
\end{eqnarray*}
and
\begin{eqnarray*}
 && D\left[\begin{array}{cc}n&n+1\\n&n+1\end{array}\right]=H_{n-1}^m, \quad
 D\left[\begin{array}{c}n\\n\end{array}\right]=-E_{n-1}^m, \\
 && D\left[\begin{array}{c}n+1\\n+1\end{array}\right]=H_{n}^m, \quad
 D\left[\begin{array}{c}n\\n+1\end{array}\right]=G_{n}^m, \quad
 D\left[\begin{array}{c}n+1\\n\end{array}\right]=-G_{n-1}^m \, .
\end{eqnarray*}

For $n=1$, (\ref{identi-E^m-E^m}) is easily verified. For $n \geq 2$, (\ref{identi-E^m-E^m}) is
the Jacobi identity (\ref{jacobi}) with $D = G_{n+1}^0$ and
\begin{eqnarray*}
 && D \left[\begin{array}{cc}1&n+1\\n&n+1\end{array}\right]
  = H_{n-1}^1 \, , \quad
 D \left[\begin{array}{c}1\\n\end{array}\right] = E_n^1 \, , \\
 && D \left[\begin{array}{c}n+1\\n+1\end{array}\right]
  = H_{n}^0 \, , \quad
 D \left[\begin{array}{c}1\\n+1\end{array}\right]
  = H_n^1 \, , \quad
 D \left[\begin{array}{c}n+1\\n\end{array}\right] = E_{n}^0 \, .
\end{eqnarray*}

(\ref{identi-F^m-F^m}) is quickly verified for $n=1,2$.
For $n \geq 3$, it is the Jacobi identity (\ref{jacobi}) with $D=H_{n+1}^0$ and
\begin{eqnarray*}
 && D \left[\begin{array}{cc}1&n+1\\n-1&n+1\end{array}\right] = G_{n-1}^1 \, , \quad
  D \left[\begin{array}{c}1\\n-1\end{array}\right] = F_n^1 \, , \\
 && D \left[\begin{array}{c}n+1\\n+1\end{array}\right] = H_{n}^0 \, , \quad
  D \left[\begin{array}{c}1\\n+1\end{array}\right] = H_n^1 \, , \quad
  D \left[\begin{array}{c}n+1\\n-1\end{array}\right] = F_{n}^0 \, .
\end{eqnarray*}
\end{proof}

\begin{lemma}
\begin{eqnarray}
    H_n^m \, (F_n^m + E_n^m) = (G_n^m)^2 + H_{n+1}^m H_{n-1}^m \, , \quad n=1,2,\ldots \, , \quad m=0,1 .
       \label{FGH_identity}
\end{eqnarray}
\end{lemma}
\begin{proof} For $n=2,3,\ldots$, let us introduce
\begin{eqnarray*}
S_n^m=\left|
\begin{array}{cccc}
c_m&\ldots&c_{m+n-2}&c_{m+n}\\
\vdots&\ddots&\vdots&\vdots\\
c_{m+n-2}&\ldots&c_{m+2n-4}&c_{m+2n-2}\\
c_{m+n}&\ldots&c_{m+2n-2}&c_{m+2n}
\end{array}\right| \, .
\end{eqnarray*}
Moreover, we set $S_0^0=u_0(u_{-1}+u_0)$, $S_0^1=u_0u_1$ and $S^m_1=c_{m+2}$. The stated equation is then obtained
by eliminating $S_n^m$ from the following two equations,
\begin{eqnarray*}
 && S_n^m H_n^m - (G_n^m)^2 - H_{n+1}^m H_{n-1}^m = 0 \, , \qquad n=1,2,\ldots\, ,  \\
 && F_n^m + E_n^m = S_n^m \, , \qquad n=0,1,2,\ldots \, ,
\end{eqnarray*}
which we shall prove now. For $n=1$, the first equation is easily verified.
For $n \geq 2$, it is the Jacobi identity (\ref{jacobi}) with $D=H_{n+1}^m$ and
\begin{eqnarray*}
&& D\left[\begin{array}{cc}n&n+1\\n&n+1\end{array}\right]=H_{n-1}^m, \quad
   D\left[\begin{array}{c}n\\n\end{array}\right]=S_n^m, \\
&& D\left[\begin{array}{c}n+1\\n+1\end{array}\right]=H_{n}^m, \quad
   D\left[\begin{array}{c}n\\n+1\end{array}\right]=D\left[\begin{array}{c}n+1\\n\end{array}\right]=G_{n}^m.
\end{eqnarray*}
For $n=0,1,2$, the second equation can be checked directly. For $n\geq3$, let
\begin{eqnarray*}
    f(x) = \sum_{k=0}^{+\infty} \frac{c_k}{k!} x^k \, , \qquad
    \sigma_n^m = \det\left(f^{(i+j+m)}\right)_{0\leq i,j\leq n-1} \, ,
\end{eqnarray*}
where $f^{(k)} = \rmd^k f/\rmd x^k$. Then we have $c_k = f^{(k)}(0)$ and
\begin{small}
\begin{eqnarray*}
\fl \frac{\rmd^2\sigma_n^m}{\rmd x^2} &=&
 \left|
\begin{array}{cccc}
f^{(m)}&\cdots&f^{(m+n-2)}&f^{(m+n)}\\
\vdots&\ddots&\vdots&\vdots\\
f^{(m+n-2)}&\cdots&f^{(m+2n-4)}&f^{(m+2n-2)}\\
f^{(m+n)}&\cdots&f^{(m+2n-2)}&f^{(m+2n)}
\end{array}\right|
 = \left|
\begin{array}{cccc}
f^{(m)}&\cdots&f^{(m+n-2)}&f^{(m+n+1)}\\
\vdots&\ddots&\vdots&\vdots\\
f^{(m+n-1)}&\cdots&f^{(m+2n-3)}&f^{(m+2n)}
\end{array}\right|  \\
\fl && + \left|
\begin{array}{ccccc}
f^{(m)}&\cdots&f^{(m+n-3)}&f^{(m+n-1)}&f^{(m+n)}\\
\vdots&\ddots&\vdots&\vdots&\vdots\\
f^{(m+n-1)}&\cdots&f^{(m+2n-4)}&f^{(m+2n-2)}&f^{(m+2n-1)}
\end{array}\right| \, .
\end{eqnarray*}
\end{small}
Setting $x=0$ in the above expression, we find that $S_n^m=F_n^m+E_n^m$ is also true for $n \geq 3$.
Also see \cite{CCSHY15} for a similar argument.
\end{proof}

\begin{lemma}
For $m=0,1$, and $n=1,2,\ldots$,
\begin{eqnarray}
 && \frac{G_{n}^m}{H_{n}^m}-\frac{G_{n-1}^m}{H_{n-1}^m}
  = u_{2n-2+m}+u_{2n-1+m} \, ,  \label{G_n-G_{n-1}=u_n} \\
 && \frac{G_{n}^m}{H_{n}^m} = \sum_{i=1}^{2n-1+m}u_i \, , \label{G^m-infty} \\
 &&\frac{E_n^m}{H_n^m}-\frac{E_{n-1}^m}{H_{n-1}^m}
  = \frac{G_n^m}{H_n^m}(u_{2n-2+m}+u_{2n-1+m})+u_{2n-1+m}u_{2n+m} \, , \label{Key-E^0-recursion02}\\
 &&\frac{F_n^m}{H_n^m}-\frac{F_{n-1}^m}{H_{n-1}^m}
  = \frac{G_{n-1}^m}{H_{n-1}^m}(u_{2n-2+m}+u_{2n-1+m})-u_{2n-3+m}u_{2n-2+m} \, . \label{Key-F^1-recursion02}
\end{eqnarray}
\end{lemma}
\begin{proof}
First we note that (\ref{G^1_identity}), (\ref{G^0_identity}) and (\ref{u<-H}) imply (\ref{G_n-G_{n-1}=u_n}),
using $G_0^0 =-u_0$, $G_0^1=0$, $H_0^m=1$, $m=0,1$.
Summing (\ref{G_n-G_{n-1}=u_n}) at lattice sites $1$ to $n$, leads to (\ref{G^m-infty}).

 From the identity (\ref{iden for F^m}), we obtain
\begin{eqnarray*}
&&\frac{F_n^m}{H_n^m}+\frac{E_{n-1}^m}{H_{n-1}^m}=\frac{G_n^mG_{n-1}^m}{H_n^mH_{n-1}^m} \, , \qquad n=1,2,\ldots \, .
\end{eqnarray*}
Using (\ref{FGH_identity}) to eliminate either $F^m_n$ or $E^m_{n-1}$ in this equation, leads to the two equations
\begin{eqnarray*}
&& \frac{E_n^m}{H_n^m}-\frac{E_{n-1}^m}{H_{n-1}^m}
 = \frac{G_n^m}{H_n^m} \Big( \frac{G_n^m}{H_n^m}-\frac{G_{n-1}^m}{H_{n-1}^m} \Big)
   + \frac{H_{n+1}^mH_{n-1}^m}{(H_n^m)^2} \, , \qquad n=1,2,\ldots, \\
 && \frac{F_n^m}{H_n^m}-\frac{F_{n-1}^m}{H_{n-1}^m}
 = \frac{G_{n-1}^m}{H_{n-1}^m} \Big( \frac{G_n^m}{H_n^m}-\frac{G_{n-1}^m}{H_{n-1}^m} \Big)
   - \frac{H_{n}^mH_{n-2}^m}{(H_{n-1}^m)^2} \, , \qquad n=2,3,\ldots \, .
\end{eqnarray*}
Now we use (\ref{G_n-G_{n-1}=u_n}) and (\ref{u<-H}) to conclude that (\ref{Key-E^0-recursion02}) holds
for $n=1,2,\ldots$, and (\ref{Key-F^1-recursion02}) holds for $n=2,3,\ldots$.
Using the definitions of $H_0^m$, $G_0^m$, $F_0^m$, $F_1^m$, $m=0,1$, one easily verifies that
(\ref{Key-F^1-recursion02}) also holds for $n=1$.
\end{proof}

\begin{lemma}
 For $m=0,1$, and $n=1,2,\ldots$, we have
\begin{eqnarray}
 \frac{E_{n}^m-F_{n}^m}{H_{n}^m}-\frac{E_{n-1}^m-F_{n-1}^m}{H_{n-1}^m}
  &=& (u_{2n-2+m}+u_{2n-1+m})^2 + u_{2n-3+m}u_{2n-2+m} \nonumber \\
  && + u_{2n-1+m}u_{2n+m} \, , \label{recurrence-E-F} \\
 \frac{E_{n}^0-F_{n}^0}{H_{n}^0}-\frac{E_{n-1}^1-F_{n-1}^1}{H_{n-1}^1}
  &=& u_{2n-1} ( u_{2n-2} + u_{2n-1} + u_{2n} ) \, , \label{new-E^0-F^0-E^1+F^1}  \\
 \frac{E_{n}^1-F_{n}^1}{H_{n}^1}-\frac{E_{n}^0-F_{n}^0}{H_{n}^0}
  &=& u_{2n} ( u_{2n-1} + u_{2n} + u_{2n+1} ) \, .  \label{new-E^1-F^1-E^0+F^0}
\end{eqnarray}
\end{lemma}
\begin{proof}
Taking the difference between (\ref{Key-E^0-recursion02}) and (\ref{Key-F^1-recursion02}),
and using (\ref{G_n-G_{n-1}=u_n}), leads to (\ref{recurrence-E-F}).
(\ref{new-E^0-F^0-E^1+F^1}) is quickly verified for $n=1$. For $n \geq 2$, we start
from the trivial identity
\begin{eqnarray*}
\fl \frac{E_{n}^0-F_{n}^0}{H_{n}^0}-\frac{E_{n-1}^1-F_{n-1}^1}{H_{n-1}^1}
 &=& \frac{E_{n}^0}{H_{n}^0}-\frac{E_{n-1}^0}{H_{n-1}^0}
   +\frac{E_{n-1}^0H_{n-1}^1-H_{n-1}^0E_{n-1}^1}{H_{n-1}^0H_{n-1}^1} \\
 && - \Big(\frac{F_{n}^0}{H_{n}^0}-\frac{F_{n-1}^0}{H_{n-1}^0} \Big)
    - \frac{F_{n-1}^0H_{n-1}^1-F_{n-1}^1H_{n-1}^0}{H_{n-1}^0H_{n-1}^1} \, .
\end{eqnarray*}
Using (\ref{identi-E^m-E^m}), (\ref{identi-F^m-F^m}), (\ref{Key-E^0-recursion02}) and
(\ref{Key-F^1-recursion02}) with $m=0$, and (\ref{u<-H}), we obtain
\begin{eqnarray*}
\fl \frac{E_{n}^0-F_{n}^0}{H_{n}^0}-\frac{E_{n-1}^1-F_{n-1}^1}{H_{n-1}^1}
 &=& \Big( \frac{G_n^0}{H_n^0}-\frac{G_{n-1}^0}{H_{n-1}^0} \Big) (u_{2n-2}+u_{2n-1})
     + u_{2n-3}u_{2n-2}+u_{2n-1}u_{2n} \\
 &&  - u_{2n-2} \, \Big( \frac{G_n^0}{H_n^0}-\frac{G_{n-1}^1}{H_{n-1}^1}
   + \frac{G_{n-1}^1}{H_{n-1}^1}-\frac{G_{n-2}^1}{H_{n-2}^1} \Big) \, , \qquad n=2,3,\ldots \, .
\end{eqnarray*}
Applying (\ref{G^0_identity}), (\ref{G_n-G_{n-1}=u_n}) and (\ref{u<-H}) again, shows that
(\ref{new-E^0-F^0-E^1+F^1}) holds for $n=2,3,\ldots$.

With the help of (\ref{identi-E^m-E^m}), (\ref{identi-F^m-F^m}) and (\ref{u<-H}), for $n=1,2,\ldots$ we have
\begin{eqnarray*}
 \frac{E_{n}^1-F_{n}^1}{H_{n}^1}-\frac{E_{n}^0-F_{n}^0}{H_{n}^0}
 &=& \frac{E_{n}^1}{H_{n}^1}-\frac{E_{n}^0}{H_{n}^0} - \Big( \frac{F_{n}^1}{H_{n}^1}-\frac{F_{n}^0}{H_{n}^0} \Big)\\
 &=& \Big( \frac{G_{n+1}^0}{H_{n+1}^0}-\frac{G_{n}^0}{H_{n}^0}+\frac{G_{n}^0}{H_{n}^0}
     -\frac{G_{n-1}^1}{H_{n-1}^1} \Big) \, u_{2n} \, .
\end{eqnarray*}
Then, employing (\ref{G^0_identity}), (\ref{G_n-G_{n-1}=u_n}) with $m=0$, and (\ref{u<-H}),
we confirm (\ref{new-E^1-F^1-E^0+F^0}).
\end{proof}

\subsection{Proof of Theorem~\ref{thm:V1}}
In this subsection, a dot stands for a derivative with respect to $t_1$.\footnote{In the following subsections,
a dot means a derivative with respect to the ``time'' variable of the equation under consideration.}

\begin{lemma}
\label{lem:dotH_V1}
For $m=0,1$, and $n=1,2,\ldots$, (\ref{V1_c_eqs}) implies
\begin{eqnarray}
 \dot{H}_n^m &=& G_n^m - (2n-1+m) \, u_0 \, H_n^m \, . \label{dotH_V1}
\end{eqnarray}
\end{lemma}
\begin{proof}
As a consequence of (\ref{V1_c_eqs}), we have
\begin{eqnarray*}
 \dot{A}_j = A_{j+1}-\frac{u_0}{c_0}B_j  \, , \qquad j=0,1,\ldots \, ,
\end{eqnarray*}
where $A_j$ and $B_j$ have been defined in (\ref{A,B}).
This in turn implies
\begin{eqnarray*}
 \dot{H}_{n}^m &=& \sum_{j=m}^{n-1+m} \left| A_m,\ldots, A_{j-1},\dot{A}_j, A_{j+1}, \ldots, A_{n-1+m} \right| \\
 &=& G_{n}^m - \frac{u_0}{c_0}\sum_{j=m}^{n-1+m} \left| A_m,\ldots, A_{j-1},B_j, A_{j+1}, \ldots, A_{n-1+m} \right| \, .
\end{eqnarray*}
(\ref{dotH_V1}) now follows by use of (\ref{sum|ABA|}).
\end{proof}

\begin{corollary}
\label{coro:V1_Hb}
As a consequence of (\ref{V1_c_eqs}), the following identities hold for $n=1,2,\ldots$,
\begin{eqnarray*}
  && \dot{H}_{n+1}^0 H_n^1 - H_{n+1}^0 \dot{H}_n^1 = H_n^0 H_{n+1}^1 - u_0 \, H_{n+1}^0 H_n^1 \, , \\
  && H_n^0 \dot{H}_{n}^1 -  \dot{H}_n^0 H_{n}^1 = H_{n+1}^0 H_{n-1}^1 - u_0 \, H_{n}^0 H_n^1 \, .
\end{eqnarray*}
\end{corollary}
\begin{proof}
This follows immediately from the identities (\ref{G^1_identity}) and (\ref{G^0_identity}),
together with Lemma~\ref{lem:dotH_V1}.
\end{proof}

This corollary expresses a bilinearization of the Volterra lattice equation (\ref{V1}), see \ref{app:HBL}.

For $n>2$ in (\ref{V1}), Theorem~\ref{thm:V1} is now a consequence of (\ref{derivation_u}) and the preceding corollary.
For $n=1,2$, (\ref{V1}) is easily verified directly.

\subsection{Proof of Theorem~\ref{thm:naV1}}

\begin{lemma}
For $m=0,1$, and $n=1,2,\ldots$, (\ref{naV1_c_eqs}) implies
\begin{eqnarray}
 \dot{H}_n^m = (2n-1+m) \, G_n^m + (2n-1+m) \, u_0 \, H_n^m \, .  \label{dotH_naV1}
\end{eqnarray}
\end{lemma}
\begin{proof}
As a consequence of (\ref{naV1_c_eqs}), we have
\begin{eqnarray*}
  \dot{A}_j = (j+n) A_{j+1} - A_{j+1}^\ast + \frac{u_0}{c_0}B_j \, , \qquad j=0,1,\ldots \, .
\end{eqnarray*}
This implies
\begin{eqnarray*}
 \dot{H}_{n}^m &=& \sum_{j=m}^{n-1+m} \left| A_m,\ldots, A_{j-1},\dot{A}_j, A_{j+1}, \ldots, A_{n-1+m} \right| \\
 &=& (2n-1+m) G_n^m - \sum_{j=m}^{n-1+m} \left| A_m,\ldots, A_{j-1},A_{j+1}^\ast, A_{j+1}, \ldots, A_{n-1+m} \right| \\
 && + \frac{u_0}{c_0} \sum_{j=m}^{n-1+m} \left| A_m,\ldots, A_{j-1},B_j, A_{j+1}, \ldots, A_{n-1+m} \right| \, .
\end{eqnarray*}
Now we use (\ref{sum|AA*_j+1A|}) and (\ref{sum|ABA|}) to obtain (\ref{dotH_naV1}).
\end{proof}

\begin{corollary}
\label{coro:naV1_Hb}
For $n=1,2,\ldots$, (\ref{naV1_c_eqs}) implies
\begin{eqnarray*}
 && \dot{H}_{n+1}^0 H_n^1 - H_{n+1}^0 \dot{H}_n^1 = 2n \, H_n^0 H_{n+1}^1 + G_{n+1}^0 H_n^1
    + u_0 \, H^0_{n+1} H^1_n  \, , \\
 && H_n^0 \dot{H}_{n}^1 - \dot{H}_n^0 H_{n}^1 = (2n-1) H_{n+1}^0 H_{n-1}^1 + H_n^0 G_n^1 + u_0 \, H_{n}^0 H_n^1 \, .
\end{eqnarray*}
\end{corollary}
\begin{proof}
This follows from the preceding lemma, together with (\ref{G^1_identity}) and (\ref{G^0_identity}).
\end{proof}

For $n>2$ in (\ref{naV1}), Theorem~\ref{thm:naV1} follows from (\ref{derivation_u}), (\ref{G_n-G_{n-1}=u_n}),
and the preceding corollary.
For $n=1,2$, (\ref{naV1}) is easily verified directly.

\subsection{Proof of Theorem~\ref{thm:V2}}
\label{subsec:V2}

\begin{lemma}
\label{lem:dotH_V2}
For $m=0,1$, and $n=1,2,\ldots$, (\ref{V2_c_eqs}) implies
\begin{eqnarray}
\fl \dot{H}_n^m = E_n^m - F_n^m - (n-1+m)(u_0+u_{-1}) \, u_0 \, H_n^m - (2n-2+m) \, u_0 \, u_1 \, H_n^m \, .
   \label{dotH_V2}
\end{eqnarray}
\end{lemma}
\begin{proof}
(\ref{V2_c_eqs}) implies
\begin{eqnarray*}
 \dot{A}_j = A_{j+2} - \frac{(u_0+u_{-1}) u_0}{c_0} (B_j - c_0 A_j) - \frac{u_0}{c_0} \hat{B}_j
  \, , \qquad j=0,1,\ldots \, .
\end{eqnarray*}
This in turn leads to
\begin{eqnarray*}
\fl \dot{H}_{n}^m &=& \sum_{j=m}^{n-1+m}\left|A_m,\ldots, A_{j-1},\dot{A}_j, A_{j+1}, \ldots, A_{n-1+m} \right| \\
\fl &=& E_{n}^m - F_{n}^m - \frac{(u_0+u_{-1})u_0}{c_0} \sum_{j=m}^{n-1+m}
        \left|A_m,\ldots, A_{j-1}, B_j, A_{j+1}, \ldots, A_{n-1+m} \right| \\
\fl && + n (u_0+u_{-1}) u_0 H^m_n
    - \frac{u_0}{c_0} \sum_{j=m}^{n-1+m} \left|A_m,\ldots, A_{j-1},\hat{B}_j, A_{j+1}, \ldots, A_{n-1+m}\right| \, .
\end{eqnarray*}
Now (\ref{dotH_V2}) is obtained by use of (\ref{sum|ABA|}), (\ref{sum|AhatBA|}) and $u_1 = c_1/c_0$.
\end{proof}

\begin{corollary}
\label{coro:V2_Hb}
For $n=1,2,\ldots$, (\ref{V2_c_eqs}) implies
\begin{eqnarray*}
\fl \dot{H}_{n+1}^0 H_n^1 - \dot{H}_n^1 H_{n+1}^0
  = (E_{n+1}^0 - F_{n+1}^0) H_n^1 - (E_n^1-F_n^1) H_{n+1}^0 - u_0 u_1 H_{n+1}^0 H_n^1 \, , \\
\fl \dot{H}_{n}^1 H_n^0 - \dot{H}_n^0 H_{n}^1
  = (E_n^1-F_n^1) H_{n}^0 - (E_{n}^0-F_{n}^0) H_n^1 - (u_{-1}+u_0+u_1) u_0 H_{n}^0 H_n^1 \, .
\end{eqnarray*}
\end{corollary}
\begin{proof}
These are simple consequences of the preceding lemma.
\end{proof}

\paragraph{Proof of Theorem~\ref{thm:V2}.}
(\ref{V2}) is easily verified for $n=1,2$. We use (\ref{derivation_u}) and the preceding
corollary to find
\begin{eqnarray*}
 \dot{u}_{2n}
 = u_{2n} \Big[ \frac{E_{n+1}^0-F_{n+1}^0}{H_{n+1}^0} - \frac{E_{n}^0-F_{n}^0}{H_{n}^0}
     - \Big( \frac{E_{n}^1-F_{n}^1}{H_{n}^1}-\frac{E_{n-1}^1-F_{n-1}^1}{H_{n-1}^1} \Big) \Big]
\end{eqnarray*}
and
\begin{eqnarray*}
 \dot{u}_{2n-1}
 = {u}_{2n-1} \Big[ \frac{E_n^1-F_n^1}{H_n^1}-\frac{E_{n-1}^1-F_{n-1}^1}{H_{n-1}^1}
   - \Big( \frac{E_n^0-F_n^0}{H_n^0}-\frac{E_{n-1}^0-F_{n-1}^0}{H_{n-1}^0} \Big) \Big] \, ,
\end{eqnarray*}
for $n=2,3,\ldots$.
Now (\ref{recurrence-E-F}) shows that the second Volterra lattice hierarchy equation is also
satisfied for $n=3,4,\ldots$.

\subsection{Proof of Theorem~\ref{thm:naV2}}
\label{subsec:naV2}

\begin{lemma}\label{lem-dot-H^0-H^1-2nd-iso-LV}
For $m=0,1$, and $n=1,2,\ldots$, (\ref{naV2_c_eqs}) implies
\begin{eqnarray}
\fl \dot{H}_n^m &=& - (2n-2+m) F_n^m + (2n+m) E_n^m + 2(n-1+m) (u_0+u_{-1}) \, u_0 \, H_n^m \nonumber \\
\fl             && + (2n-2+m) \, u_0 \, u_1 \, H_n^m \, . \label{dotH_naV2}
\end{eqnarray}
\end{lemma}
\begin{proof}
As a consequence of (\ref{naV2_c_eqs}), we have
\begin{eqnarray*}
\fl  \dot{A}_j=(j+n+1)A_{j+2}-A_{j+2}^\ast+\frac{2(u_0+u_{-1})u_0}{c_0} (B_j - c_0 A_j) + \frac{u_0}{c_0}\hat{B}_j
  \, , \qquad j=0,1,\ldots \, ,
\end{eqnarray*}
which implies
\begin{eqnarray*}
\fl \dot{H}_{n}^m &=& \sum_{j=m}^{n-1+m} \left|A_m,\ldots, A_{j-1},\dot{A}_j, A_{j+1}, \ldots, A_{n-1+m} \right|
      \nonumber\\
\fl &=& (2n+m) E_{n}^m - (2n-1+m) F_{n}^m - \sum_{j=m}^{n-1+m}
     \left|A_m,\ldots, A_{j-1},A_{j+2}^\ast, A_{j+1}, \ldots, A_{n-1+m} \right|\nonumber\\
\fl && +\frac{2(u_0+u_{-1}) u_0}{c_0}\sum_{j=m}^{n-1+m}
    \left|A_m,\ldots, A_{j-1},B_j, A_{j+1}, \ldots, A_{n-1+m} \right|\nonumber\\
\fl && - 2n (u_0+u_{-1}) u_0 H^m_n
      + \frac{u_0}{c_0}\sum_{j=m}^{n-1+m} \left|A_m,\ldots, A_{j-1},\hat{B}_j, A_{j+1}, \ldots, A_{n-1+m} \right| \, .
\end{eqnarray*}
Now we use (\ref{sum|AA*A|}), (\ref{sum|ABA|}) and
(\ref{sum|AhatBA|}) to obtain (\ref{dotH_naV2}).
\end{proof}

\begin{corollary}
\label{coro:naV2_Hb}
For $n=1,2,\ldots$, (\ref{naV2_c_eqs}) implies
\begin{eqnarray*}
 \dot{H}_{n+1}^0 H_n^1 - \dot{H}_n^1 H_{n+1}^0
    &=& (2n+2) (E_{n+1}^0-F_{n+1}^0) H_n^1 - 2n (E_n^1-F_n^1) H_{n+1}^0   \\
     && + 2 F_{n+1}^0 H_n^1 - (E_n^1+F_n^1) H_{n+1}^0 + u_0 u_1 H_n^1 H_{n+1}^0 \, , \\
   \qquad \dot{H}_{n}^1 H_n^0 - \dot{H}_n^0 H_{n}^1
    &=& 2n (E_n^1-F_n^1) H_{n}^0 - 2n (E_{n}^0 - F_{n}^0) H_n^1 - 2 F_{n}^0 H_n^1 \\
     && + (E_n^1+F_n^1) H_{n}^0 + (2u_0+2u_{-1}+u_1) u_0 H_n^1 H_n^0 \, .
\end{eqnarray*}
\end{corollary}
\begin{proof}
These are simple consequences of the preceding lemma.
\end{proof}

\paragraph{Proof of Theorem~\ref{thm:naV2}.}
For $n=1,2$, (\ref{naV2}) is quickly verified.
Using (\ref{derivation_u}) and the preceding corollary, we find
\begin{eqnarray*}
\fl \dot{u}_{2n}
 = u_{2n} \Big[ (2n+2) \Big( \frac{E_{n+1}^0-F_{n+1}^0}{H_{n+1}^0}-\frac{E_{n}^0-F_{n}^0}{H_{n}^0} \Big)
    -2n \Big( \frac{E_{n}^1-F_{n}^1}{H_{n}^1}-\frac{E_{n-1}^1-F_{n-1}^1}{H_{n-1}^1} \Big) \\
\fl \quad + 2 \Big( \frac{E_{n}^0-F_{n}^0}{H_{n}^0}-\frac{E_{n-1}^1-F_{n-1}^1}{H_{n-1}^1} \Big)
    + 2 \Big( \frac{F_{n+1}^0}{H_{n+1}^0}-\frac{F_{n}^0}{H_{n}^0} \Big)
    - \Big( \frac{E_{n}^1}{H_{n}^1}-\frac{E_{n-1}^1}{H_{n-1}^1} \Big)
    - \Big( \frac{F_{n}^1}{H_{n}^1}-\frac{F_{n-1}^1}{H_{n-1}^1} \Big) \Big]
\end{eqnarray*}
and
\begin{eqnarray*}
\fl \dot{u}_{2n-1}
 = {u}_{2n-1} \Big[ 2n \Big( \frac{E_n^1-F_n^1}{H_n^1}-\frac{E_{n-1}^1-F_{n-1}^1}{H_{n-1}^1} \Big)
     -2n \Big( \frac{E_n^0-F_n^0}{H_n^0}-\frac{E_{n-1}^0-F_{n-1}^0}{H_{n-1}^0} \Big) \\
\fl \quad +2 \Big( \frac{E_{n-1}^1-F_{n-1}^1}{H_{n-1}^1}-\frac{E_{n-1}^0-F_{n-1}^0}{H_{n-1}^0} \Big)
    -2 \Big( \frac{F_n^0}{H_n^0}-\frac{F_{n-1}^0}{H_{n-1}^0} \Big) 
 +  \frac{E_{n}^1}{H_{n}^1}-\frac{E_{n-1}^1}{H_{n-1}^1}
    +  \frac{F_{n}^1}{H_{n}^1}-\frac{F_{n-1}^1}{H_{n-1}^1} \Big] \, ,
\end{eqnarray*}
for $n=2,3,\ldots$.
Using (\ref{G^m-infty})-(\ref{recurrence-E-F}) in both expressions, (\ref{new-E^0-F^0-E^1+F^1}) in the
first and (\ref{new-E^1-F^1-E^0+F^0}) in the second, we see
that (\ref{naV2}) also holds for $n=3,4,5,\ldots$.

\setcounter{equation}{0}

\section{Conclusions and final remarks}
\label{sec:conclusions}
Via expressing solutions in terms of Hankel determinants, we achieved a transformation of the first
two autonomous and also the non-autonomous flows of the extended Volterra lattice hierarchy to Riccati
systems. Since the latter are known to be linearizable, we thus achieved a linearization of these
nonlinear partial differential-difference equations. This resolves corresponding results
in \cite{Bere+Shmo94} from the restriction to the boundary condition $u_0=0$.

Originally we had some hope to be able to extend the results of this work to the whole extended
Volterra lattice hierarchy. However, a corresponding treatment of higher than second equations
of the extended Volterra lattice hierarchy meets with rapidly increasing complexity,
and the underlying structure is not yet visible.

Let $u_{-2},u_{-1},u_0,u_1$ be smooth functions of $t$. Let $c_j$, $j=0,1,\ldots$, satisfy
\begin{eqnarray*}
     \frac{\rmd c_j}{\rmd t_3} &=& c_{j+3} - \frac{u_0}{c_0} \Big( \sum_{i=2}^j c_i c_{j+2-i}
     + (u_{-1}+u_0) \, \sum_{i=1}^j c_i c_{j+1-i} \\
    && + (u_{-2} u_{-1} + u_{-1}^2 + 2 u_{-1} u_0 + u_0^2 + u_0 u_1) \, \sum_{i=0}^{j-1} c_i c_{j-i} \Big) \, ,
      \qquad j > 0 \, ,  \\
     \frac{\rmd c_0}{\rmd t_3} &=& c_3 + u_0 u_1 \, c_1 \, .
\end{eqnarray*}
Then computer algebra computations suggest that (\ref{u<-H}) determines a solution of the third
Volterra lattice equation, which is (\ref{eqV^k}) with $k=3$, on the right half lattice.
For the $k$-th flow we thus expect $k$ sums of quadratic terms in the evolution equations
for the functions $c_j$. In the step from $k$ to $k+1$, this means inclusion of products of
two variables at a more remote distance on the lattice.

The linearizations of equations of the extended Volterra lattice hierarchy are actually
obtained via a bilinearization, as an intermediate step. For the most prominent member (\ref{V1}),
this is expressed in Corollary~\ref{coro:V1_Hb} and further explained in \ref{app:HBL}.
For the first non-autonomous flow, a bilinearization is given by Corollary~\ref{coro:naV1_Hb},
together with (\ref{G^1_identity}) and (\ref{G^0_identity}). Corresponding bilinearizations of
the second autonomous and non-autonomous flows have been obtained in Corollaries~\ref{coro:V2_Hb}
and \ref{coro:naV2_Hb}, together with the respective equations stated in Lemma~\ref{lem:det-identities}
and (\ref{FGH_identity}).

Can we say something about regularity of solutions? This is subtle.
According to a theorem of Hamburger \cite{Hamb20}, a sequence $s_j$, $j = 0,1,\ldots$,
can be represented as a sequence of moments (see (\ref{moments_s})),
with a positive measure $\mu$ on the real line, if and only if the Hankel matrices
$\hat{\mathcal{H}}_n = (s_{i+j})$, where $i,j=0,\ldots,n-1$, $n=1,2,\ldots$, are positive
semi-definite \cite{Hamb20}. Moreover, if the support of $\mu$ contains infinitely many points,
then $\det\hat{\mathcal{H}}_n >0$, for $n=1,2,\ldots$ \cite{Simon98,Berg+Szwa15}. If this is so,
then it implies, via (\ref{T->V:tH->H}), regularity of a corresponding solution of an equation from
the extended Volterra lattice hierarchy.

If the functions $a_j$, $j=0,1,\ldots$, in (\ref{P_recur}) are positive and bounded, the
spectral measure of the (then selfadjoint) operator (Jacobi matrix) $L$ has the required
properties \cite{Bere+Shmo94}, so that a solution $u_n$ of the form (\ref{u<-H}) is regular.
This can be achieved at least at some value of time, say $t=0$. It is then shown in \cite{Bere+Shmo94}
that the time evolution of the measure, induced by (\ref{spec_evolution}), leads to a
solution of the corresponding equation of the extended Volterra lattice hierarchy on a
finite time interval and on the right half lattice, with boundary data $u_0=0$.\footnote{This is
done in \cite{Bere+Shmo94} moreover for an extension of the Toda lattice hierarchy.}

In our work, we characterized solutions on the right half lattice in terms of boundary data, a
point of view also taken in \cite{PSZ07}.
To determine those boundary data that correspond to regular solutions is an open problem, which
might be solvable using the results and tools of \cite{Bere+Shmo94}, or other methods from the
theory of orthogonal polynomials. 

Besides the structural insights, in particular those expressed in our main results, the Riccati 
systems, shown to be equivalent to equations of the extended Volterra lattice hierarchy, provide 
us with an unfamiliar approach to exact solutions of the latter. Note that here we describe solutions 
in terms of determinants of (Hankel) matrices the size of which grows with the lattice site number. 
Since we reach all solutions, this raises the question how e.g. soliton solutions can be characterized in 
this way.  
The (integrable) equations of the extended Volterra lattice hierarchy possess solutions 
outside the familiar families of solitons and algebro-geometric (periodic) solutions, and we expect 
our results to be of help to reveal them. This still needs elaboration, partly following \cite{PSZ07}. 
All this should then also have some impact on scientific problems modeled by the Volterra lattice 
equation or its companions. 

Our results should also be useful to explore the structure of reductions of the extended Volterra lattice 
hierarchy. In this context we should mention the application of the Hankel determinant formula for the 
Toda lattice in order to prove the generic polynomiality of $\tau$ functions of the Painlev\'e equations
\cite{KMNOY01,KMO07}.

\vspace{.3cm}

\noindent
\textbf{Acknowledgment.} X.-M.~Chen has been supported by the Sino-German (CSC-DAAD)
Postdoc Scholarship Program 2016 and the DAAD Research Grants - Short-Term Grants 2018 (57378443). 
She would like to thank A.S.~Zhedanov for a very
enlightening correspondence concerning his work \cite{PSZ07}, and X.K.~Chang for many
very helpful discussions.
X.-B.~Hu has been also partially supported
by the National Natural Science Foundation of China (Grant no. 11331008, 11571358). 
Last but not least we have to thank an Editor-in-Chief for several suggestions that 
improved this work. 

\begin{appendix}

\section{From Toda to Volterra}
\label{app:T->V}
The Toda lattice equation in Flaschka variables reads
\begin{eqnarray*}
     \dot{a}_n = \frac{1}{2} a_n \, (b_{n+1} - b_n) \, , \qquad
     \dot{b}_n = a_n^2 - a_{n-1}^2 \, .
\end{eqnarray*}
Its inverse spectral problem (see \cite{Bere+Shmo94} and references cited there) can be solved via
\begin{eqnarray*}
    a_n = \frac{\sqrt{\hat{H}_n \hat{H}_{n+2}}}{\hat{H}_{n+1}}
\end{eqnarray*}
and a corresponding expression for $b_n$, see \cite{Bere+Shmo94}. Here $\hat{H}_n$
is the determinant of the Hankel matrix $\hat{\mathcal{H}}_n = (s_{i+j})$, where $i,j=0,\ldots,n-1$,
built with the moments
\begin{eqnarray}
     s_j = \int_\bbR \lambda^j \, d\mu(\lambda) \, , \qquad \quad j=0,1,\ldots \, ,  \label{moments_s}
\end{eqnarray}
where $\mu$ is an infinite positive measure on the real line. 
On the level of moments, the time dependence corresponds to a simple deformation of the measure.  
This translates the Toda lattice equation into a Riccati system for the moments \cite{Bere+Shmo94,PSZ07}. 

The reduction to the Volterra lattice involves choosing an \emph{even} measure, so that all odd moments 
are zero. Setting
\begin{eqnarray*}
    c_j := s_{2j} \, ,
\end{eqnarray*}
by exchanges of rows and columns we find that
\begin{eqnarray}
    \hat{H}_{2n} = H^0_n \, H^1_n \, , \qquad
    \hat{H}_{2n+1} = H^0_{n+1} \, H^1_n \, .    \label{T->V:tH->H}
\end{eqnarray}
In terms of $u_n = a_{n-1}^2$, we thus obtain
\begin{eqnarray*}
   u_{2n-1} = \frac{\hat{H}_{2(n-1)} \hat{H}_{2n}}{(\hat{H}_{2n-1})^2}
            = \frac{H^0_{n-1} H^1_n}{H^0_n \, H^1_{n-1}} \, , \quad
   u_{2n} = \frac{ \hat{H}_{2n-1} \hat{H}_{2n+1} }{(\hat{H}_{2n})^2}
           = \frac{H^0_{n+1} H^1_{n-1}}{H^0_n \, H^1_n} \, .
\end{eqnarray*}
These are the expressions in (\ref{u<-H}). Also see \cite{TNI01,Chu08}.

\section{Another expression for the extended Volterra lattice hierarchy}
\label{app:BS}
In \cite{Bere+Shmo94} (equation (7.14) therein), the following generalization of the Volterra
lattice equation appeared,
\begin{eqnarray}
 \dot{a}_n(t) &=& \{\Phi(L,t)\}_{n+1,n} + \frac{1}{2}a_n(\{\Psi(L,t)\}_{n+1,n+1}
                  - \{\Psi(L,t)\}_{n,n}) \nonumber \\
 && + a_{n+1} \{\Phi(L,t)D_{L}\}_{n+2,n} - a_{n-1} \{\Phi(L,t)D_{L}\}_{n+1,n-1} \, ,
     \label{Volterra_BS}
\end{eqnarray}
where $a_{-1}=0$ and $n=0,1,\ldots$. Here $\Phi$ and $\Psi$ are polynomials, in a parameter $\lambda$, of the form
\begin{eqnarray*}
 \Phi(\lambda,t)=\sum_{i=0}^l\varphi_{2i+1}(t) \, \lambda^{2i+1} \, , \qquad
 \Psi(\lambda,t)=\sum_{i=0}^m\psi_{2i}(t) \, \lambda^{2i} \, ,
\end{eqnarray*}
with coefficients that are functions of $t$.
Furthermore, $L$ is the Jacobi matrix with the only non-zero entries
given by
\begin{eqnarray*}
     L_{j+1,j} = a_j = L_{j,j+1} \, , \qquad \quad j=0,1,\ldots \, .
\end{eqnarray*}
Let $P_k(\lambda,t)$, $k=0,1,\ldots$, be a sequence of functions (actually symmetric orthogonal polynomials),
satisfying the recurrence relation
\begin{eqnarray}
  a_{k-1} P_{k-1} + a_k P_{k+1} = \lambda \, P_{k} \, , \qquad \quad k=0,1,\ldots \, ,   \label{P_recur}
\end{eqnarray}
where $P_{-1}=0$. We define $d_{jk}$, $k=1,2,\ldots$, $j=0,1,\ldots,k-1$, by
\begin{eqnarray*}
  \frac{\pa P_k}{\pa \lambda} = \sum_{j=0}^{k-1} d_{jk} P_j \, , \qquad \quad k=1,2,\ldots \, .
\end{eqnarray*}
Then $D_L$ is the strictly upper triangular matrix with entries $d_{jk}$.
These can be computed recursively using the following equation, which is obtained by differentiation of (\ref{P_recur})
with respect to $\lambda$, and using the preceding relation,
\begin{eqnarray*}
\fl && (a_{k-1} d_{0,k-1} + a_k d_{0,k+1} - a_0 d_{1,k}) \, P_0
    + \sum_{j=1}^{k-2} ( a_{k-1} d_{j,k-1} + a_k d_{j,k+1} -a_j d_{j+1,k} - a_{j-1} d_{j-1,k} ) \, P_j \\
\fl && + (a_k d_{k-1,k+1} - a_{k-2} d_{k-2,k}) \, P_{k-1} + (a_k d_{k,k+1} - a_{k-1} d_{k-1,k} - 1) \, P_{k} = 0 \, .
\end{eqnarray*}
In particular, we find
\begin{eqnarray*}
\fl  d_{k,k+1} = \frac{k+1}{a_k} \, , \quad
 d_{k,k+2} = 0 \, , \quad
 d_{k,k+3} = \frac{2 \sum_{i=0}^k a_i^2 - (k+1) a_{k+1}^2}{a_k a_{k+1} a_{k+2}} \, , \quad k=0,1,\ldots \, .
\end{eqnarray*}
As a consequence of the recurrence relation (\ref{P_recur}), the Jacobi matrix $L$ represents the
operator of multiplication by $\lambda$ in the Hilbert space spanned by $\{ P_k \,|\, k=0,1,\ldots \}$,
whereas $D_L$ represents the operator $\pa/\pa \lambda$.
$L$ is a Lax operator for (\ref{Volterra_BS}), and we have non-isospectrality if $\Phi \neq 0$, since
\begin{eqnarray}
     \frac{\rmd \lambda}{\rmd t} = \Phi(\lambda,t) \, .   \label{spec_evolution}
\end{eqnarray}
We refer to \cite{Bere+Shmo94} for details.

The Volterra lattice equation (\ref{V1}) and the second hierarchy flow (\ref{V2}) are obtained from
(\ref{Volterra_BS}) by setting $\Phi =0$ and $\Psi(\lambda) = \lambda^2$, respectively 
$\Psi(\lambda) = \lambda^4$, and using the transformation
\begin{eqnarray}
      u_n = a_{n-1}^2 \, , \qquad \quad n=0,1,\ldots \, .  \label{a->u_transf}
\end{eqnarray}
Note that the condition $a_{-1}=0$ enforces $u_0=0$.
The first non-autonomous flow (\ref{naV1}) is obtained in the same way by choosing
$\Phi(\lambda) = \frac{1}{2} \lambda^3$ and $\Psi = \lambda^2$.
\vspace{.2cm}

\noindent\textbf{Example.}
Let $\Phi(\lambda) = \frac{1}{2} \lambda^5$ and $\Psi = 2 \lambda^4$. Then (\ref{Volterra_BS}) leads to
\begin{eqnarray*}
\fl \dot{a}_n &=& \frac{1}{2} a_n \Big( (2-n) a_{n-2}^2 a_{n-1}^2 -n a_{n-1}^4 + (1-n) a_{n-1}^2 a_n^2
               + a_n^4 + (4+n) a_n^2 a_{n+1}^2 \\
\fl && + 2 a_{n-1}^2 a_{n+1}^2 + (n+3) a_{n+1}^4 + (n+3) a_{n+1}^2 a_{n+2}^2
    + 2 (a_{n+1}^2 - a_{n-1}^2) \sum_{i=0}^{n-2} a_i^2 \Big) \, .
\end{eqnarray*}
In terms of (\ref{a->u_transf}), we obtain (\ref{naV2}).
\hfill $\Box$

It is quite evident now that the autonomous Volterra lattice hierarchy is contained in (\ref{Volterra_BS})
with $\Phi=0$, choosing for $\Psi$ the members of the sequence of even powers of $\lambda$. The non-autonomous equations
are recovered, up to addition of multiples of the right hand side of autonomous flows, if we choose $\Psi=0$
and for $\Phi(\lambda)$ the members of the sequence of odd powers of $\lambda$.

\section{Hirota bilinearization of the Volterra lattice equation}
\label{app:HBL}
 From \cite{Nari82} (equation (22) therein, with $k=2$, $m=1$ and $t \mapsto -t$), we recall that, via
\begin{eqnarray*}
 u_n = \frac{{\uptau}_{n-\frac{3}{2}} \, \uptau_{n+\frac{3}{2}}}{\uptau_{n-\frac{1}{2}} \, \uptau_{n+\frac{1}{2}}} \, ,
\end{eqnarray*}
the following (Hirota) bilinearization of the Volterra lattice equation (\ref{V1}) is obtained,
\begin{eqnarray*}
\fl 2 \, \sinh(D_n) \Big\{ \Big[ \sinh(\frac{1}{2} D_n) D_t + 2 \, \sinh(D_n) \sinh(-\frac{1}{2} D_n) \Big]
   \uptau_n \cdot \uptau_n \Big\}
   \cdot \Big( \cosh(\frac{1}{2} D_n) \uptau_n \cdot \uptau_n \Big) = 0 \, .
\end{eqnarray*}
Here $D_n$ and $D_t$ are Hirota bilinear operators \cite{Hiro04}. The latter equation is equivalent to
\begin{eqnarray*}
\fl    \Big[ \sinh(\frac{1}{2} D_n) D_t + 2 \, \sinh(D_n) \sinh(-\frac{1}{2}D_n) \Big] \uptau_n \cdot \uptau_n
  = \chi \, \cosh(\frac{1}{2} D_n) \, \uptau_n \cdot \uptau_n \, ,
\end{eqnarray*}
where $\chi$ is an arbitrary function of $t$, independent of $n$. Shifting by half a lattice spacing,
and writing $\chi = 1 - u_0$, with a function $u_0(t)$, this can be expressed as
\begin{eqnarray*}
   \dot{\uptau}_{n+1} \, \uptau_n - \uptau_{n+1} \, \dot{\uptau}_n
   = \uptau_{n-1} \, \uptau_{n+2} - u_0 \, \uptau_{n+1} \, \uptau_{n} \, .
\end{eqnarray*}
Via
\begin{eqnarray*}
   \uptau_{2n} = H_n^0 \, ,  \qquad \uptau_{2n-1} = H_{n-1}^1 \, ,
\end{eqnarray*}
this is turned into the equations in Corollary~\ref{coro:V1_Hb}. This is related to the two-field form of
the Volterra lattice equation, see \cite{Hiro+Sats76b,Maru+Ma02,Suri03}, for example.

In \cite{Nari82}, the freedom expressed by the function $\chi$ is not taken into consideration. Equation
(23) in \cite{Nari82} corresponds to the case where $\chi=0$ (also see, e.g., \cite{Veks04}).
In the present context, where we express $\uptau$-functions in terms of Hankel determinants, we see
that the additional freedom is necessary. In \cite{Maru+Ma02} (see (1.5) therein), for example, a bilinearization of
the two-field form of (\ref{V1}) is given, which is of the form of the equations in Corollary~\ref{coro:V1_Hb},
but with constant coefficient of the last term.

\section{Thiele expansion and the first non-autonomous Volterra flow}
\label{app:algo}
Let us consider the algorithm \cite{brezinski201002cross}
\begin{eqnarray*}
  \gamma_{-1}(t)=0 \, , \quad \gamma_{0}(t) = f(t) \, ,  \qquad
  \gamma_{n+1}(t) = \gamma_{n-1}(t) + \frac{g_n(t)}{\dot {\gamma}_n(t)} \qquad n=0,1,\ldots \, ,
\end{eqnarray*}
where $f(t)$ is a smooth function and $g_n(t)$, $n=0,1,\ldots$, are given functions, assumed to be
nowhere vanishing. Via the Miura transformation (cf. \cite{brezinski2010,brezinski201002cross} in the present context)
\begin{eqnarray*}
  u_n(t) = - M_{n-1}(t) \, M_{n}(t) \, , \qquad  M_{n}(t)=\dot{\gamma}_{n}(t)/g_{n}(t) \, , \qquad
     n = 1,2,\ldots \, ,
\end{eqnarray*}
this implies
\begin{eqnarray}
  && \dot{u}_1=u_1\left( g_2 u_2+(g_1-g_0)u_1 \right),\label{1st-flow-iso-noniso-u1} \nonumber \\
  && \dot{u}_n =u_n \left( g_{n+1} u_{n+1} + (g_{n}-g_{n-1}) u_n - g_{n-2} u_{n-1} \right) \qquad n=2,3,\ldots, \qquad
  \label{1st-flow-iso-noniso}
\end{eqnarray}
and we have to set $u_{0}=0$, so that these two equations fit together in the sense that the second
equation can be extended to $n=1$ and then includes the first equation. Extending the Miura transformation to $n=0$,
thus introducing $g_{-1}$ (which does not appear in the algorithm), $u_0=0$ is a consequence.

If $g_n(t) = 1$ for all $n$, this is the Volterra
lattice equation (\ref{V1}), in which case
the above algorithm is known as ``confluent $\varepsilon$-algorithm'' \cite{wynn1960a,wynn1960b}.
If $g_n(t)=n$, (\ref{1st-flow-iso-noniso}) is the first non-autonomous Volterra lattice equation (\ref{naV1}).
If $g_n(t) = n+1$, (\ref{1st-flow-iso-noniso}) is a combination of the Volterra lattice
and the first non-autonomous Volterra lattice equation. More precisely, the right hand side is then
$ V^{(1)} + \mathcal{V}^{(1)}$.
In this case we are dealing with the ``confluent form of the $\rho$-algorithm'' \cite{wynn1960a},
\begin{eqnarray*}
 \rho_{-1}(t)=0 \, , \quad
  \rho_{0}(t) = f(t) \, ,  \quad
 \rho_{n+1}(t) = \rho_{n-1}(t)+\frac{n+1}{\dot {\rho}_n(t)} \qquad n=0,1,\ldots \, .
\end{eqnarray*}
In particular, this computes the functions in \emph{Thiele's expansion formula}, a continued fraction expansion
(see e.g. \cite{brezinski1991,cuyt1987nonlinear}, also for the notation)
of $f(t+h)$ around $t$,
\begin{eqnarray*}
  f(t+h) = f(t) + \frac{\qquad h\qquad|}{|\rho_1(t)-\rho_{-1}(t)}+\frac{\qquad h\qquad|}{|\rho_2(t)-\rho_{0}(t)}
    + \frac{\qquad h\qquad|}{|\rho_3(t)-\rho_{1}(t)}+\ldots \, .
\end{eqnarray*}

Moreover, for $g_{n}=\alpha n+\beta$, where $\alpha$ and $\beta$ are functions of $t$, the $\gamma_n$, defined via
the above recurrence relation, have explicit expressions in terms of Hankel determinants,
\begin{eqnarray}
   \gamma_{2k}(t)=\frac{\tilde{H}_{k+1}^0(t)}{\tilde{H}_{k}^2(t)} \, , \qquad
   \gamma_{2k+1}(t)=\frac{\tilde{H}_{k}^3(t)}{\tilde{H}_{k+1}^1(t)} \, ,
     \label{determ-solu-conf-epsion}
\end{eqnarray}
where
\begin{eqnarray*}
  && \tilde{H}_0^m=1 \, , \quad
   \tilde{H}_n^m = \det(\xi_{i+j+m})_{i,j=0}^{n-1} \, ,  \quad m = 0,1,2,\ldots \, ,  \\
  && \quad \xi_0 = f(t) \, , \quad
   \dot{\xi}_j(t) = g_{j}(t)\xi_{j+1}(t) \, , \quad j = 0,1,2,\ldots \, .
\end{eqnarray*}
Note that $\xi_{j+1} = c_j$, $j=0,1,\ldots$, with the $c_j$ used elsewhere in this work. By using the Miura transformation, 
one can also obtain the determinant expressions (\ref{u<-H}) for $u_n$ from (\ref{determ-solu-conf-epsion}).

\vskip.1cm
\noindent
\textbf{Remark D.1.}
Our results motivate the following generalization of the $\gamma$-algorithm,
\begin{eqnarray}
&& \gamma_{-1}(t)=0 \, , \quad
  \gamma_0(t)=f(t) \, , \nonumber \\
&& \gamma_{n+1}(t) = \gamma_{n-1}(t)+\frac{g_n(t)}{\dot{\gamma}_{n}(t)-\psi(t)},\quad n=0,1,\ldots \, .
   \label{conf-epsilon-algor-conv}
\end{eqnarray}
For $g_{n}=\alpha n+\beta$, with functions $\alpha(t)$ and $\beta(t)$, the $\gamma_n$ are still given by
the expressions in (\ref{determ-solu-conf-epsion}), but now with
\begin{eqnarray}
 && \xi_0=f(t) \, ,\quad
    \dot{\xi}_0(t) = g_0(t) \, \xi_1(t) + \psi(t) \, , \nonumber\\
 &&\dot{\xi}_j(t)=g_{j}(t)\xi_{j+1}(t)-\psi(t)\, \sum_{i=1}^j \xi_i \xi_{j+1-i} \, ,
   \quad j=1,2,3,\ldots \, \label{evo-xi-convo}.
\end{eqnarray}
Using the Miura transformation, but now with $M_n(t) = \left( \dot{\gamma}_{n}(t)-\psi(t) \right)/g_{n}(t)$, we are
again led to the second equation of (\ref{1st-flow-iso-noniso}), which is accompanied by 
\begin{eqnarray*}
  \dot{u}_1 = u_1 \left( g_2 u_2+(g_1-g_0)u_1 -\frac{\psi}{\gamma_1} \right) \, .
\end{eqnarray*}
Let us set $\psi(t) = g_{-1}(t) \, \gamma_1(t) \, u_0(t)$, with arbitrary $u_0(t)$. Recall that 
$\xi_{j+1} = c_j$ and note that $\gamma_1 = 1/\xi_1 = 1/c_0$ by use of (\ref{determ-solu-conf-epsion}). 
If $g_n(t)=1$, then $\psi(t) = u_0(t)/c_0(t)$ and the last equation of (\ref{evo-xi-convo}), expressed in terms of $c_j$,  
coincides with (\ref{V1_c_eqs}). This makes contact with Theorem~\ref{thm:V1}.
If $g_n(t)=n$, then $\psi(t)=-u_0(t)/c_0(t)$ and the last equation of (\ref{evo-xi-convo}) coincides with (\ref{naV1_c_eqs}). 
Here we meet the case described in Theorem~\ref{thm:naV1}.

If $g_n(t)=n+1$, we have $g_{-1} = 0$, which is indeed necessary in order to cast the above equation for $\dot{u}_1$
and the second of (\ref{1st-flow-iso-noniso}) into the combination of autonomous and non-autonomous
Volterra flows. However, this enforces $\psi = 0$, so that we are back to the old algorithm in this case.

The generalized algorithm (\ref{conf-epsilon-algor-conv}) still needs further exploration. It is of
interest since its solutions still admit Hankel determinant representations.

\end{appendix}

\end{document}